\documentclass[aps,pre,twocolumn,groupedaddress,showpacs]{revtex4-1}
\usepackage{graphicx}
\usepackage{subfig}
\usepackage{enumerate}
\usepackage{multirow}
\usepackage{amsmath,amsthm,amssymb}

\newtheorem{theorem}{Theorem}

\renewenvironment{proof}[1][Proof:]
{\begin{trivlist}\item[\hskip \labelsep {\itshape {\bfseries #1}}] }
{\qed\end{trivlist}}

\begin{document}


    \title{The success of complex networks at criticality}
    

	\author{Victor Hernandez-Urbina}
	\email[]{j.v.hernandez-urbina@ed.ac.uk}
	\author{J. Michael Herrmann}
	\affiliation{Institute of Perception, Action and Behaviour, University of Edinburgh. EH8 9AB \\
		United Kingdom.}
	\author{Tom L. Underwood}
	\affiliation{School of Physics and Astronomy, SUPA, University of Edinburgh. EH9 3JZ \\
		United Kingdom.}


	\date{\today}

	\begin{abstract}
			In spiking neural networks an action potential could in principle trigger subsequent spikes
			in the neighbourhood of the initial neuron.
			A successful spike is that which trigger subsequent spikes giving rise to cascading behaviour within the system.
			In this study we introduce a metric to assess the success of spikes emitted by integrate-and-fire 
			neurons arranged in complex topologies and whose collective behaviour is undergoing a phase transition
			that is identified by neuronal avalanches
			that become clusters of activation whose distribution of sizes can be approximated by a power-law.
			In numerical simulations we report that scale-free networks with the small-world property is the structure in which
			neurons possess more successful spikes.
			As well, we conclude both analytically and in numerical simulations that fully-connected networks are structures
			in which neurons perform worse.
			Additionally, we study how the small-world property affects spiking behaviour and its success in scale-free networks.
	\end{abstract}


	\maketitle

	\section{Introduction}
		\label{intro}
		Cascading behaviour in complex networks refers to a domino effect resulting from the activation
		of nodes in a network whose internal dynamics are subject to threshold mechanisms
		and propagation of events (e.g. action potentials, diseases, fads, articles becoming cited, etc.).
		In the context of cascading behaviour, the success of a node refers to its capacity (once
		it becomes active as a result of its internal dynamics) to trigger subsequent activations
		in its neighbourhood.
		Intuitively, network structure plays a major role on determining how activity spreads through a system.
		In this respect, the discovery of topological features such as the small-world property and the scale-invariance
		of the degree distribution in many real-world networks provided a new perspective in which to analyze
		cascading activity.
		
		The study of complex networks took off when it was observed that
		the essence of real-world networks cannot be captured by the
		random network model introduced by Erd\"os and Renyi~\cite{newman2010networks}
		nor by regular structures such as lattices.
		The Watts and Strogatz model~\cite{watts1998collective} was proposed to describe
        		a class of networks that lie halfway between randomness and regularity.
		This class of networks are characterized by a small average shortest path length (a feature observed
		in random networks) and an average clustering coefficient significantly larger than 
		expected by chance (a feature observed in regular lattices). 
		Taken together, these properties offer a structural benefit to the processes taking place 
		within the network, such as optimal information transmission 
		that results from speeding up the communication among otherwise distant nodes.
		A term that summarizes the presence of these two properties is that of 
		the \emph{small-world property}.
		Networks that exhibit the small-world property
		are so diverse and can be found in social, technological 
		and biological contexts, to name a few~\cite{humphries2008network}.

		Scale-invariance in the distribution of the node degrees of a network is a phenomenon 
		observed in real-world networks. 
        		Networks that exhibit this particular feature are known as \emph{scale-free networks}.
		In this type of networks the 
		probability $P(k)$ that a node connects to $k$ other nodes follows a power-law 
		$P(k)\sim k^{-\gamma}$~\cite{barabasi1999emergence}. It implies the existence of many poorly 
		connected nodes coexisting with very few but 
		not negligible massively connected nodes \emph{hubs}. 
		Scale-invariant degree distribution and the small-world
		property are by no means exclusive and in fact many
		scale-invariant networks are also small-world. 		

		When it comes to the dynamics of a system comprising numerous interconnected elements
		interacting non-linearly,
		a considerable number of studies have been dedicated to the occurrence of power-law 
		behaviour and its relationship to the notion of phase transitions, $1/f^{\alpha}$ noise
		and \emph{self-organized criticality} 
		(SOC)~\cite{bak1988self} that results from the collective dynamics of 
		the threshold units comprising a system.

		The concept of SOC, 
		has been suggested to explain the dynamics of phenomena as diverse as 
		plate tectonics~\cite{gutenberg1956magnitude}, piles of granular
		matter~\cite{frette1996avalanche}, forest fires~\cite{bak1990forest}, neuronal 
		avalanches~\cite{eurich2002finite} (see below), 
		and mass extinctions~\cite{bak1997nature}, among several others.
		Moreover, SOC implies the existence of a critical point that becomes an attractor in the 
		collective dynamics of a system.
		Such a critical point or regime denotes a state of the system in which the 
	    	collective dynamics are undergoing a phase transition.
		As such, it represents the boundary between two different states of the system (e.g.~order 
		and chaos) and it is identified by the presence of power laws
		in the distribution of events, the divergence of the correlation length, among 
		others~\cite{bak1988self}.
		
        		In the context of brain networks, the presence of neuronal avalanches has been
	        observed as the result of spontaneous activity in local field potentials of cultured slices
        		of rat cortex~\cite{beggs2003neuronal},
		and in the superficial cortical layers of awake, resting 
		primates~\cite{petermann2009spontaneous}.
		As the name suggests, neuronal avalanches are an example of cascading behaviour 
		triggered by spiking in groups of neurons.
	    	The observed avalanches are stable and repeatable spatiotemporal patterns of 
		activity~\cite{beggs2004neuronal}, which might relate them to memory mechanisms inside 
		the brain.
		
		In models of neuronal avalanches, it has been reported that the distribution of their sizes
		as well as the distribution of their durations can be approximated by a power 
		law with very precise exponents in the thermodynamic 
		limit as well as scaling relationships among system sizes and 
		exponents~\cite{eurich2002finite,levina2007dynamical}.
		As mentioned earlier, this phenomenon has also been reported in real brain tissue (although
		with finite-size effects)~\cite{beggs2003neuronal}. 
		Power-law behaviour in the dynamics of neuronal avalanches relate this biological process 
		to the notion of SOC described above.
		
		Critical dynamics of brain networks have been studied thoroughly in artificial models, 
		and it has been found that the critical regime implies several computational benefits for 
		the system, namely: optimal information transmission and 
	        maximum dynamic range~\cite{kinouchi2006optimal}, maximum information 
        		storage~\cite{haldeman2005critical,uhlig2013critical}, 
	        stability of information transmission~\cite{bertschinger2004real}, among others.
        		Hence the criticality hypothesis for brain dynamics, which states that neural networks 
	        operate at the \emph{edge of chaos}, that is, at the critical point in a phase transition 
	        	between total randomness and boring order~\cite{beggs2008criticality}.
		
		In this paper we study the success of integrate-and-fire nodes in terms of their capacity to trigger
		subsequent spikes in their neighbourhood once they become active, and thus giving rise to cascading activity.
		We study this local performance in topologies such as fully-connected, random, and scale-free networks
		with varying amounts of the small-world property when the systems are at the critical state of their
		collective dynamics.

	\section{Model}
		\label{model}
    		\subsection{The Eurich model}
			\label{Sect:model}
			The model consists of $N$ non-leaky
			integrate-and-fire nodes and was formulated for fully-connected networks~\cite{eurich2002finite}.
		    This model exhibits critical dynamics in the collective behaviour of the units
		    comprising the system. 
		    This fact is identified by the presence of neuronal avalanches, whose size and durations
		    can be approximated by a power law.
		    In analytical examinations, the exponents derived for such distributions are
		    $\gamma = -3/2$ and $\delta=-2$ for avalanches sizes and durations respectively,
            in the thermodynamic limit of fully-connected networks~\cite{eurich2002finite,levina2007dynamical}.		   
            Here we will extend such a model by considering also heterogeneous directed networks.
            
		    In the model, each node $j$ is characterised by a continuous
		    variable known as the membrane potential $h$, which is updated in discrete 
		    time according to the equation:

		    \begin{equation}\label{model_h}
			    h_{j}(t+1) = h_{j}(t) + \sum_{i = 1}^{N} A_{ij}w_{ij}s_{i}(t) + I_{ext}
		    \end{equation}
		
		    \noindent where $A$ denotes the asymmetric adjacency matrix with entries $A_{ij} = 1$ if 
		    node $i$ sends and edge to node $j$, and $A_{ij} = 0$ otherwise;
		    $w_{ij}$ denotes the synaptic strength from node $i$ to node $j$;
		    $s_{i}(t)\in \{1,0\}$ represents the state of node $i$ (active or quiescent, 
		    respectively) at time $t$;
		    and $I_{ext}$ denotes external input which is supplied to a node depending on the 
		    state of the system at time $t$.
		    This mechanism of external driving works as follows:
		    if there is no activity at time $t$, then a node is chosen uniformly at random and its 
		    membrane potential is increased by a fixed amount through the variable
		    $I_{ext}$.
		    If $h_{i}(t)$ exceeds the threshold $\theta$, which in simulations is set to unity,
		    then node $i$ emits a spike, which
		    changes the state of this node to active ($s_{i}(t) = 1$) and propagates its activity 
		    through its synaptic output.
		    Afterwards, the node is reset, ie. $h_{i}(t+1) = 0$.

		    
		    The coupling strength $w_{ij}$ for every node $i$ sending an edge to node $j$
		    is set according to the equation $w_{ij} = \alpha / \langle e\rangle$
		    where $\alpha$ is the control parameter of the 
		    model and $\langle e\rangle$ denotes the mean degree
		    of the network; which is the same for all the network structures considered (see
		    below), except for fully-connected networks, in which case mean connectivity is
		    $(N-1)$.
		    
		    In the first stages of our experiments we let the parameter $\alpha$ take
		    values in the interval $(0,1)$, and then we measure the deviation from the
		    best power-law fit to the distribution of avalanche sizes (see Sect.~\ref{model:implementation} for details).
		    The values of $\alpha$ for which the deviation is at its minimum are those
		    who lead the system to the critical state, and define the critical interval.
		    In a second stage of our experiments we re-start the system considering only
		    values of $\alpha$ taking uniformly at random from the critical interval.
		    The critical interval varies for different types of networks and system sizes.
		 
		 
		    As mentioned above, in this work we consider different types of network structures as
		    well as system sizes. The topologies considered are:
		    \begin{enumerate}[\it i.]
		        \item fully connected,
                \item random, 
                \item scale-free with \emph{low} mean clustering coefficient (CC) and power-law 
                in the out-degree distribution, 
		        \item scale-free with \emph{high} mean CC and power-law in out-degree 
		        distribution, 
		        \item scale-free with \emph{low} mean CC and power-law in in-degree distribution,
		        \item scale-free with \emph{high} mean CC and power-law in in-degree 
		        distribution.
		    \end{enumerate}
		 
            System sizes for each of these classes of networks are: $128$, $256$, $512$, and $1,024$.
            We should point out an important aspect of the networks that we consider.
            For the case of random and scale-free networks the number of edges is the same for each system 
            size, which results in the same average connectivity for these types of networks.
            Thus, the topology results from a particular permutation of the edges.
            However, this edge permutation is not arbitrary, but results from a particular algorithm
            depending on the structure that we want to obtain.
	        In the case of random networks the mechanism, by which we permutate such edges, is given by
	        the Erd\"os-Renyi model~\cite{newman2003structure},
	        whereas for scale-free networks with tuneable clustering we follow the ideas in Ref.~\cite{holme2002growing},
	         in which the authors present an extended version of the \emph{BA model}~\cite{barabasi1999emergence}, 
	         in which a large mean clustering coefficient is achieved
	        for scale-free networks by adding a triangle-formation step to the preferential attachment
	        algorithm. 
	        This algorithm produces scale-free networks whose mean clustering coefficients
	        match better the observations in real-world networks.
	    
		 
	        In the following section we describe the structural differences of networks whose in-degree
	        distribution follows a power-law, and networks whose out-degree distribution follows a 
	        power-law.

		 \subsection{Broadcasting hubs and absorbing hubs} 
		    Scale-free networks are characterized by a power-law approximation of their degree 
		    distribution~\cite{barabasi1999emergence}. 
		    In the case of directed networks, there are two degree distributions, one corresponding
		    to the out-degrees of nodes and another one for the in-degree of nodes.
		    In general, real-world networks are directed, and very often both their
		    degree distributions can be approximated by a power-law (e.g. the World Wide 
		    Web~\cite{newman2003structure}), or at least one of them (e.g. citation 
		    networks~\cite{newman2003structure}).
		    In either case, the presence of a long-tail in the \emph{out-degree} distribution of a 
		    network implies the existence of \emph{broadcasting hubs}, that is,
		    nodes that have massive outgoing connections compared with other nodes in the system. 
		    On the contrary, the presence of a long-tail in the 
		    \emph{in-degree distribution} implies the existence of \emph{absorbing hubs}.
		    Here, we are interested in analyzing how collective dynamics develop for the case of 
		    networks with broadcasting hubs and for networks with absorbing hubs.
		    In the following, scale-free networks with absorbing hubs will be labeled as
		    \emph{in-degree scale-free networks}, whereas those that contain broadcasting hubs
		    will be termed \emph{out-degree scale-free networks}.
		    
		    As mentioned in the introduction, the small-world property is not a binary one, and as 
		    such, there exist degrees of what we would call
		    \emph{small-world-ness}.
		    All the scale-free networks considered possess the small-world property up to a certain
		    amount.
		    In our model we consider two levels of mean clustering for scale-free networks 
		    (\emph{low} and \emph{high}) by tuning a simple parameter~\cite{holme2002growing}.
		    The process of tuning the mean clustering coefficient in these types of networks has 
		    an immediate effect on the degree of small-world-ness of such networks.
		    Scale-free networks with low mean clustering coefficient possess a low degree
		    of small-world-ness when compared against scale-free networks with high mean clustering
		    coefficient.
		    In this study, we inquire on the effects on criticality for different degrees of
		    small-world-ness.
		    
		\subsection{Node success} 
			\label{model:nodeSuccess}
		    We introduce a local measure of the performance of a node during simulation time.
		    The node success of node $i$ at time $t$ is the fraction of \emph{out-neighbors} of 
		    this node that become active at time $t+1$ when node $i$ spikes at time $t$,
		    in other words:
		    
		    \begin{equation}
			    \label{eq:phi}
			    \varphi_{i}(t) = \frac{\sum_{j = 1}^{N} A_{ij}s_{j}(t+1)}{\sum_{j = 1}^{N} A_{ij}}
		    \end{equation}
		
			\noindent where $A$ is the adjacency matrix, and $s_{j}(t+1)$ the state of node $j$
			at time $t+1$.
			
			Thus, node success measures the performance of an individual spike in terms of the
			subsequent spikes triggered by such initial activation, which occur within the
			out-neighborhood of a given node.
			In contrast to many other popular network statistics (e.g. degree distribution, 
			branching ratio~\cite{beggs2003neuronal}, etc.) node success is a local measure of performance.
			
			We consider two different averages of this measure. First, the \emph{mean node success
			per node} which results from considering only the times in which a node spikes
			and then averaging its node success at each of these times.
			Second, the \emph{mean node success per time step} which results from averaging the node successes
			of all nodes in the system at a particular time step.
			
%
%
%
%

		\subsection{Numerical implementation}
		    \label{model:implementation}
			When starting simulations, all membrane potentials are initialised at random taking
			values in the interval $(0,1)$, whereas all states are set to inactive.
			By means of external driving, activity inside the system in the form of neuronal
			avalanches is guaranteed to occur.
			However, avalanche sizes and their durations will not always be the same nor can they be
			predicted.
			
		    Both the relaxation time towards the critical state as well as the sampling time needed
		    to assess criticality depend on the system size.
			For networks consisting of $128$ nodes we allow critical dynamics to set in for one 
			million time steps according to the Eurich model~\cite{eurich2002finite}; for 
			networks of $256$ nodes we allow critical dynamics to set in for two million time steps, 
			for networks comprising $512$ we allow for three million time steps,
			and finally for networks of $1,024$ elements the dynamics run for four million
			time steps.
			This selection of times is appropriate for large events (that is avalanches that extend
			to the whole network) to take place during simulation time.
			With this in mind, we expect to have small events (i.e.~small avalanches) coexisting 
			with large events (i.e.~avalanches that span the whole system).
			An inspection of the distribution of avalanche sizes after this driving stage shows a 
			distribution that can be approximated by a power law 
			with a cut-off due to the finite nature of the system (see Sect.~\ref{results:avsizes}). 
			The power-law approximation of such a distribution implies that the system is in the 
			critical regime with very frequent small events coexisting with rare but not negligible
			large events. 
		
			We assess the quality of such a power law through the mean-squared deviation 
			$\Delta\gamma$ from the best-matching power law with exponent $\gamma$ obtained through
			regression in log-log scales.
			Our choice of using this method is due to its simplicity and justified by the 
			asymptotic unbiasedness of the estimation.
			When this error function is at its minimum, that is, when the data is best approximated
			by a power-law distribution with exponent $\gamma$, is when the system is at 
			the critical state.
			
			
			For our experiments we consider $50$ different networks per class (\emph{ii} to 
			\emph{vi} described above) and system size for the sake of statistical robustness. 
			In the case of fully-connected networks, as there exists only one fully-connected 
			network of size $N$, randomness is introduced in the seed of the pseudorandom number
			generator used in our code for each realization of the 
			experiment rather than in the structure as for the other network classes considered.
			Experiments were carried out in the EDDIE computer cluster of the University of 
			Edinburgh.
    \section{Results\label{results}}
    
        \subsection{Avalanche size distribution follows a power-law}
	        \label{results:avsizes}
    		As mentioned in Sect.~\ref{model:implementation}, 
	    	we assess the quality of the power-law approximation to the
		    distribution of avalanche sizes by estimating the deviation from the best power-law fit.
		    When such an error function reaches a minimum value of less than or equal to 0.05, we consider the 
		    event-size distribution as well approximated by a power-law and conclude that the system 
		    is in a critical state.
		    Fig.~\ref{fig1:a} shows the power-law fitting error as a function of simulation time
		    for the distribution
		    of avalanche sizes for scale-free and random networks of size $N = 512$.
		    This figure shows the deviation, $\Delta\gamma$, of our data from the best matching power 
		    law with exponent $\gamma$. 
		    In this figure, we present mean values and standard deviations of $\Delta\gamma$
		    obtained from the realizations of our experiments.
		    Fig.~\ref{fig1:b} shows the distribution of avalanches sizes for all scale-free
		    and random networks of size $N = 512$ at criticality 
		    (ie. when $\Delta\gamma\leq 0.05$ around time step $3\times10^{6}$ in Fig.~\ref{fig1:a}).
		    We show the averaged value of all realisations, but we do not present error bars
		    in Fig.~\ref{fig1:b} in order to make its presentation more accessible.
	        Although we show the distribution of avalanche sizes and the deviation from the best
	        matching power-law for a particular system size, all system sizes exhibit a similar
	        behaviour.

	    	\begin{figure}[t]
		    	\centering
			    \subfloat[Power-law fitting deviation per time step]{%
    				\includegraphics[scale=0.28]{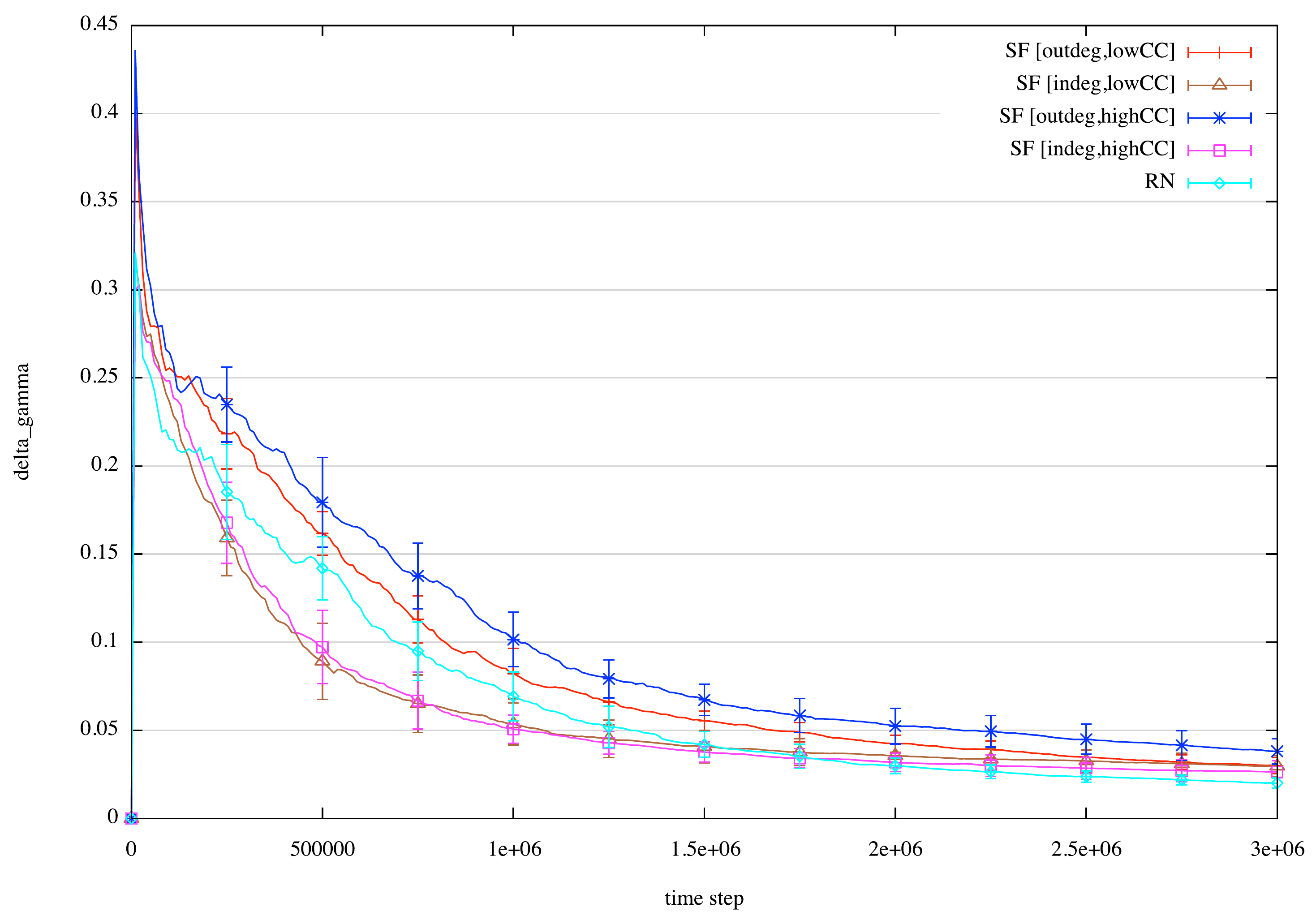}
	    			\label{fig1:a}
		    	}\\
			    \subfloat[Log-log plot of distribution of avalanche sizes $S$]{%
    				\includegraphics[scale=0.28]{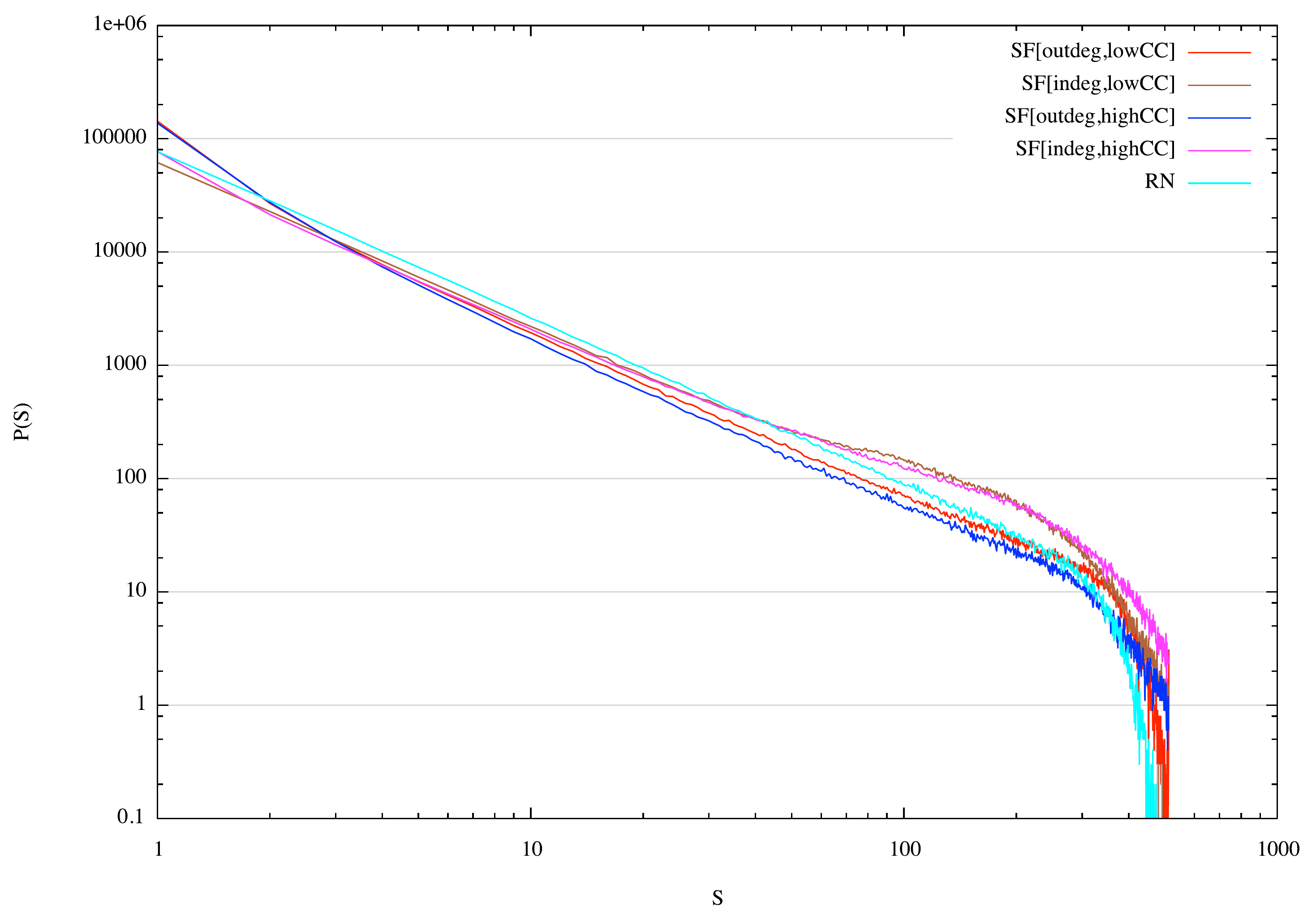}
	    			\label{fig1:b}
		    	}
    			\caption{
	    		    Deviation from power-law matching per time step for a network
		    	    of size $N = 512$. The minimum error reached at time step $3\times 10^6$
			        corresponds to the critical regime identified by a power law distribution
			        of avalanche sizes.
		            We show the averaged value of all realizations, but we do not present error bars
		            in Fig.~\ref{fig1:b} in order to make its presentation more accessible.
			    }
			    \label{fig1}
		    \end{figure}    
		    
		    Moreover, following Ref.~\cite{larremore2011predicting} we inspect the value of the largest eigenvalue
		    of the matrices $W$ associated to each network and whose entries $w_{ij}$ denote the synaptic weight
		    between node $i$ and $j$.
		    The authors in Ref.~\cite{larremore2011predicting} observe that the largest eigenvalue of the weight matrix
		    governs the dynamics of the system.
		    Through an analytical examination, it is reported that when the largest eigenvalue equals unity the system is at the
		    critical state.

	\begin{table}[t]
                \begin{tabular}{llrr}
                    \hline
                    Type & Subtype & Size & $\Lambda$\\
                    \hline
\multirow{8}{*}{Out-degree scale-free} \\                    
 & Low Mean CC & $128$ & $0.906\pm 0.029$\\
 & & $256$ & $0.9\pm 0.02$ \\
 & & $512$ & $0.95\pm 0.01$\\
 & & $1,024$ & $0.97\pm 0.008$\\
		\cline{2-4}
 & High Mean CC & $128$ & $0.89\pm 0.04$\\
 & & $256$ & $0.91\pm 0.03$\\
 & & $512$ & $0.91\pm 0.01$\\
 & & $1,024$ & $0.94\pm 0.01$\\
                    \hline 
\multirow{8}{*}{In-degree scale-free} \\ 
& Low Mean CC & $128$ & $0.98\pm 0.02$\\
 & & $256$ & $0.99\pm 0.01$\\
 & & $512$ & $1.0006\pm 0.006$\\
 & & $1,024$ & $1.001\pm 0.004$\\
		 \cline{2-4}
 & High Mean CC & $128$ & $0.96\pm 0.02$\\
 & & $256$ & $0.99\pm 0.02$\\
 & & $512$ & $1.002\pm 0.014$\\
 & & $1,024$ & $1.01\pm 0.017$\\
                     \hline
Random & & $128$ & $0.92\pm 0.012$\\
 & & $256$ & $0.99\pm 0.022$\\
 & & $512$ & $0.97\pm 0.001$\\
 & & $1,024$ & $0.98\pm 0.005$\\
 			\hline
Fully-connected & & $128$ & $0.91\pm 0.034$\\
 & & $256$ & $0.93\pm 0.0006$\\
 & & $512$ & $0.95\pm 0.001$\\
 & & $1,024$ & $0.98\pm 0.0005$ \\
 
                    \hline
                \end{tabular}
                \caption{
                		Largest eigenvalue $\Lambda$ of matrix $W$ of synaptic weights.
			It has been found analytically that $\Lambda = 1$ is associated with a system at criticality~\cite{larremore2011predicting}.
			The synaptic weight matrices of our networks have $\Lambda\approx 1$ due to finite-size effect.
			(We present mean values and standard deviations.)
                }
                \label{tabLargstEig}
	\end{table}
	
	In Table~\ref{tabLargstEig} we report the value of the largest eigenvalue $\Lambda$ of the weight matrices associated to our networks.
	In our experiments the critical state is not only identified by the power-law distribution of avalanche sizes (Fig.~\ref{fig1}) but also by the
	value close to unity of $\Lambda$.
	Due to finite size effects this value is not exactly unity but close to it.

        \subsection{Small-world property boosts network activity}
	        \label{results:spiking}
            The small-world property affects the rate of firing of nodes comprising a network.
            Fully-connected networks are structures in which all nodes exhibit a similar firing rate,
            giving rise to a well defined mean and variance (see Fig.~\ref{fig6:a})
            unlike scale-free networks in which the variance of the firing rate seemingly diverges, and
            thus its mean cannot characterize the network activity (not shown here).
            In fact, this latter type of structure contains nodes whose firing rate can far exceed the firing
            rate in fully-connected networks (see Fig.~\ref{fig3:a}).
            
            Fig.~\ref{fig6:a} shows the average of the total number of spikes over all nodes
            in fully-connected networks for all system sizes considered.
            In contrast, we will show that nodes in heterogeneous topologies can perform better;
            in this case, scale-free networks possess nodes with higher firing rates than random networks.
            It is worth mentioning that this behaviour is verified in all system sizes considered.
            
			\begin{figure}[t]
				\centering
				\subfloat[Total number of spikes]{%
					\includegraphics[scale=0.28]{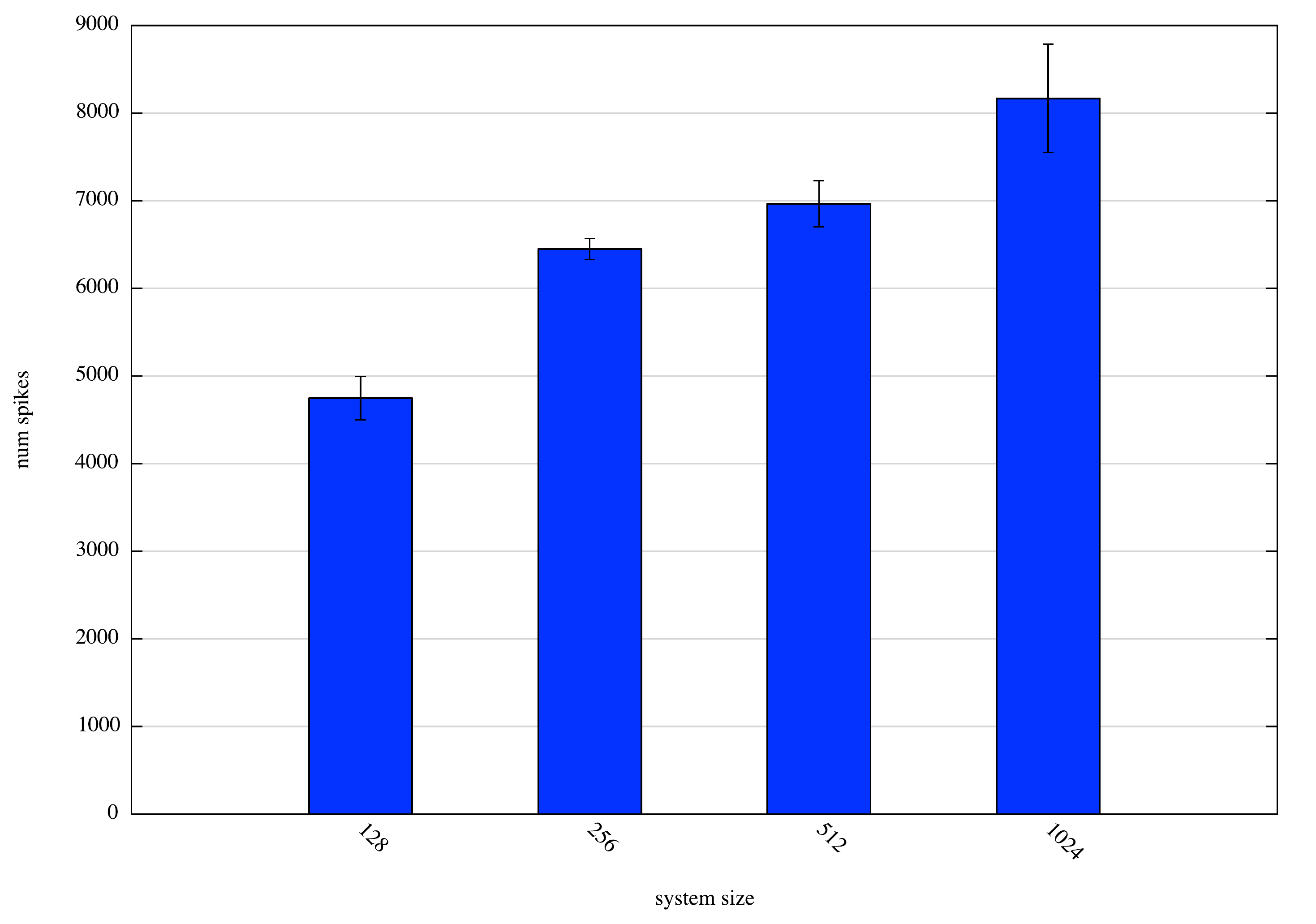}
					\label{fig6:a}
				}\\
				\subfloat[Total mean node success]{%
					\includegraphics[scale=0.28]{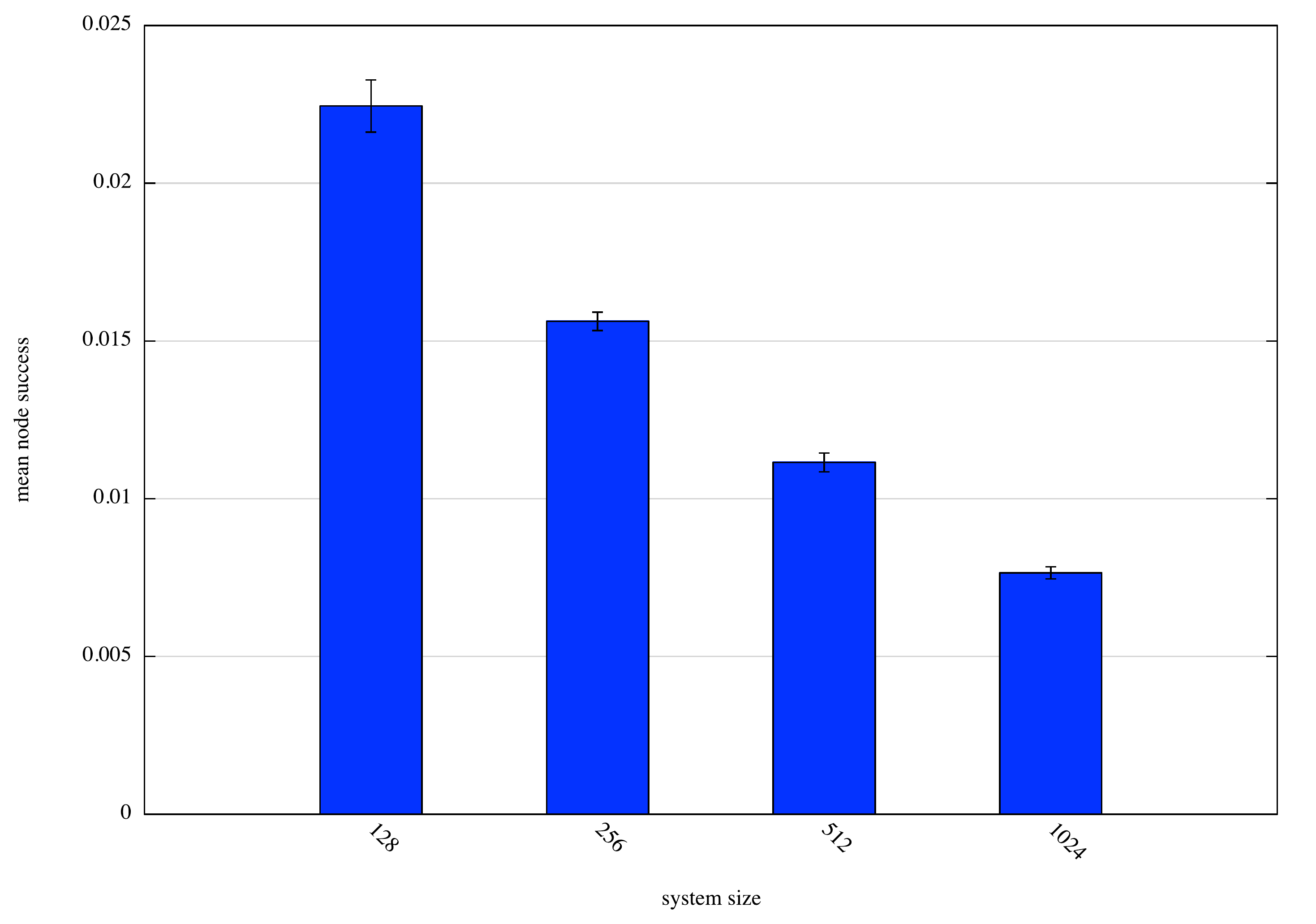}
					\label{fig6:b}
				}
				\caption{
				    Average number of spikes emitted by fully-connected networks per system size
				    and their mean success.
				    (Error bars denote standard deviations.)
				}
				\label{fig6}
			\end{figure}            
            
            We pose several questions regarding the relationship between network structure and
            dynamics.
            The first is: are the nodes with higher local CC those that
            spike more often, that is, do better-clustered nodes fire more?
            Surprisingly, nodes with low local CC exhibit a larger spiking
            rate than more clustered nodes.
            Fig.~\ref{fig3:a} shows this behaviour for networks of size $N = 1,024$.
            Here we show not only that low locally clustered nodes fire more but also that
            in in-degree scale-free networks nodes can fire more than in any other type of structure.

			\begin{figure}[t]
				\centering
				\subfloat[Total number of spikes]{%
					\includegraphics[scale=0.28]{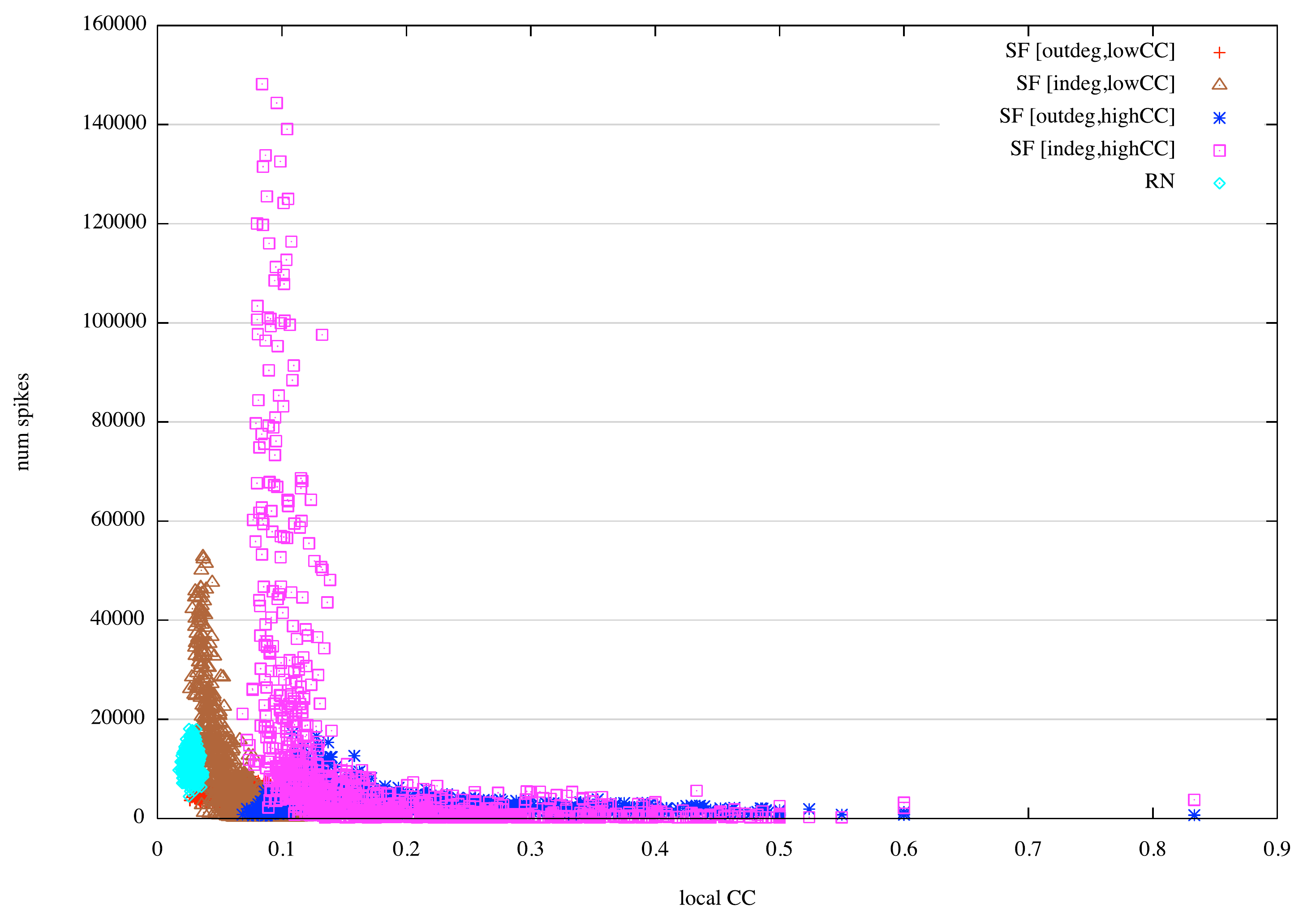}
					\label{fig3:a}
				}\\
				\subfloat[Total mean node success]{%
					\includegraphics[scale=0.28]{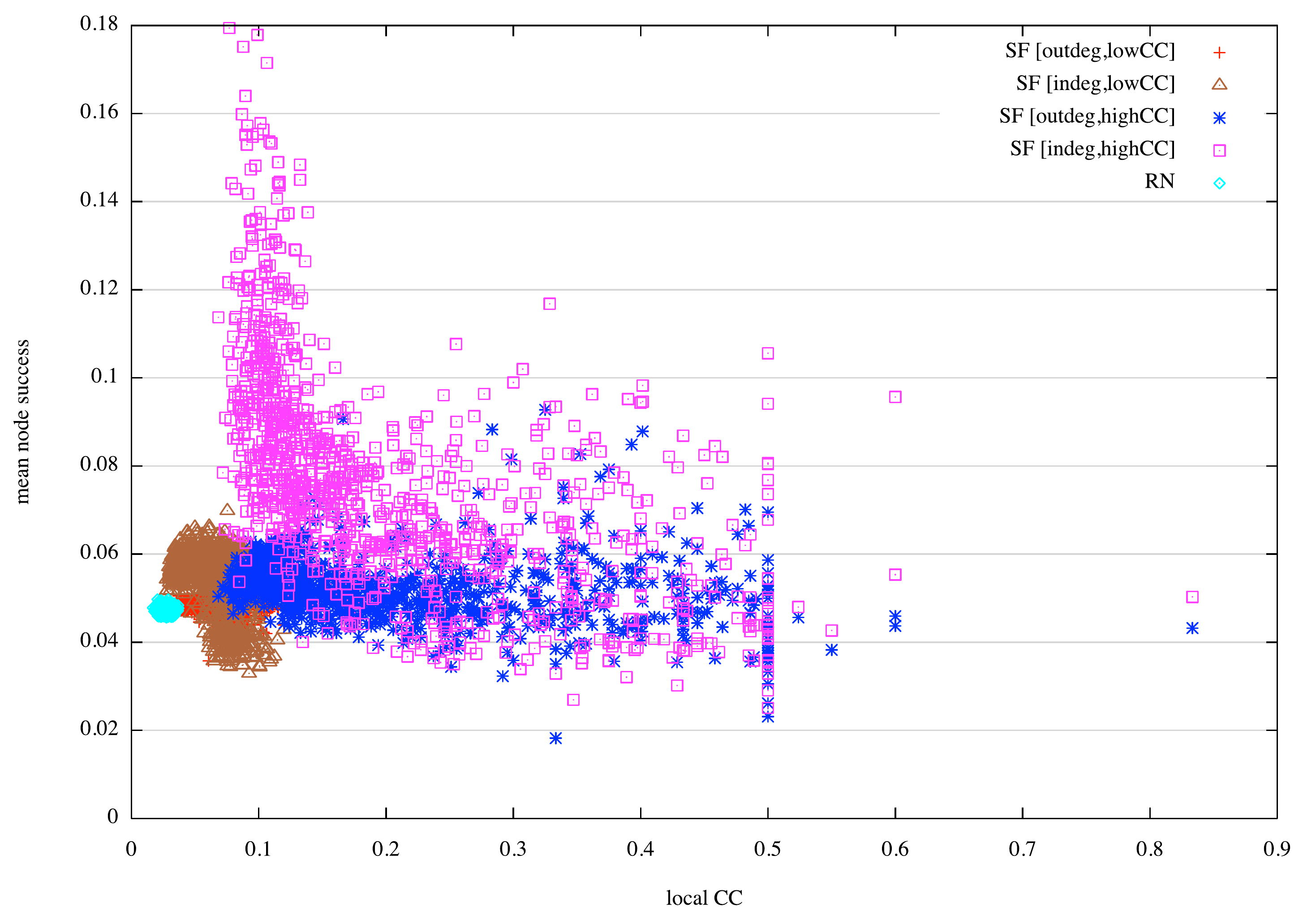}
					\label{fig3:b}
				}
				\caption{
				    Total number of spikes and mean node success per local CC
				    for heterogeneous topologies of size $N = 1,024$.
				    Low clustered nodes are responsible for spiking more frequent in
				    scale-free networks. Their spikes tend to be more successful as well.
				}
				\label{fig3}
			\end{figure}             
            
            A question that arises at this point is the following:
            are those low locally clustered nodes who spike so frequently in in-degree scale-free
            networks the absorbing hubs or any other type of node?
            Taking a look at network topology,
            we verify that indeed hubs (either absorbers or broadcasters) are in general low
            locally clustered.
            In Fig.~\ref{fig4} we present for scale-free networks of size $N = 1,024$ and two 
            levels of mean CC
            (low and high) the relationship between in-degree/out-degree and local CC.
            The more a node is in-connected the lower its local CC (likewise when reversing the 
            direction of edges, which yields broadcasting hubs).

			\begin{figure}[t]
				\centering
				\subfloat[Low mean CC]{%
					\includegraphics[scale=0.28]{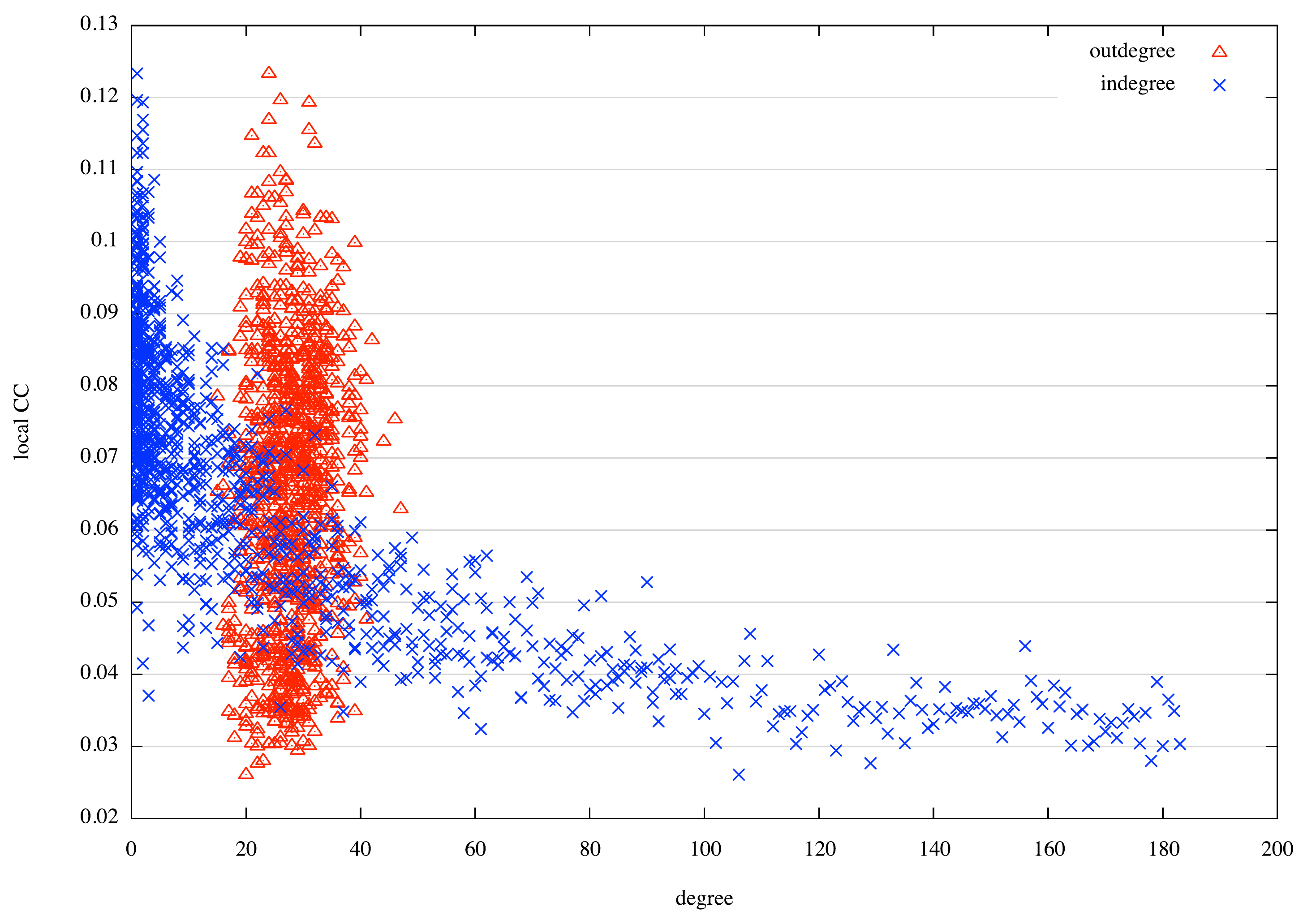}
					\label{fig4:a}
				}\\
				\subfloat[High mean CC]{%
					\includegraphics[scale=0.28]{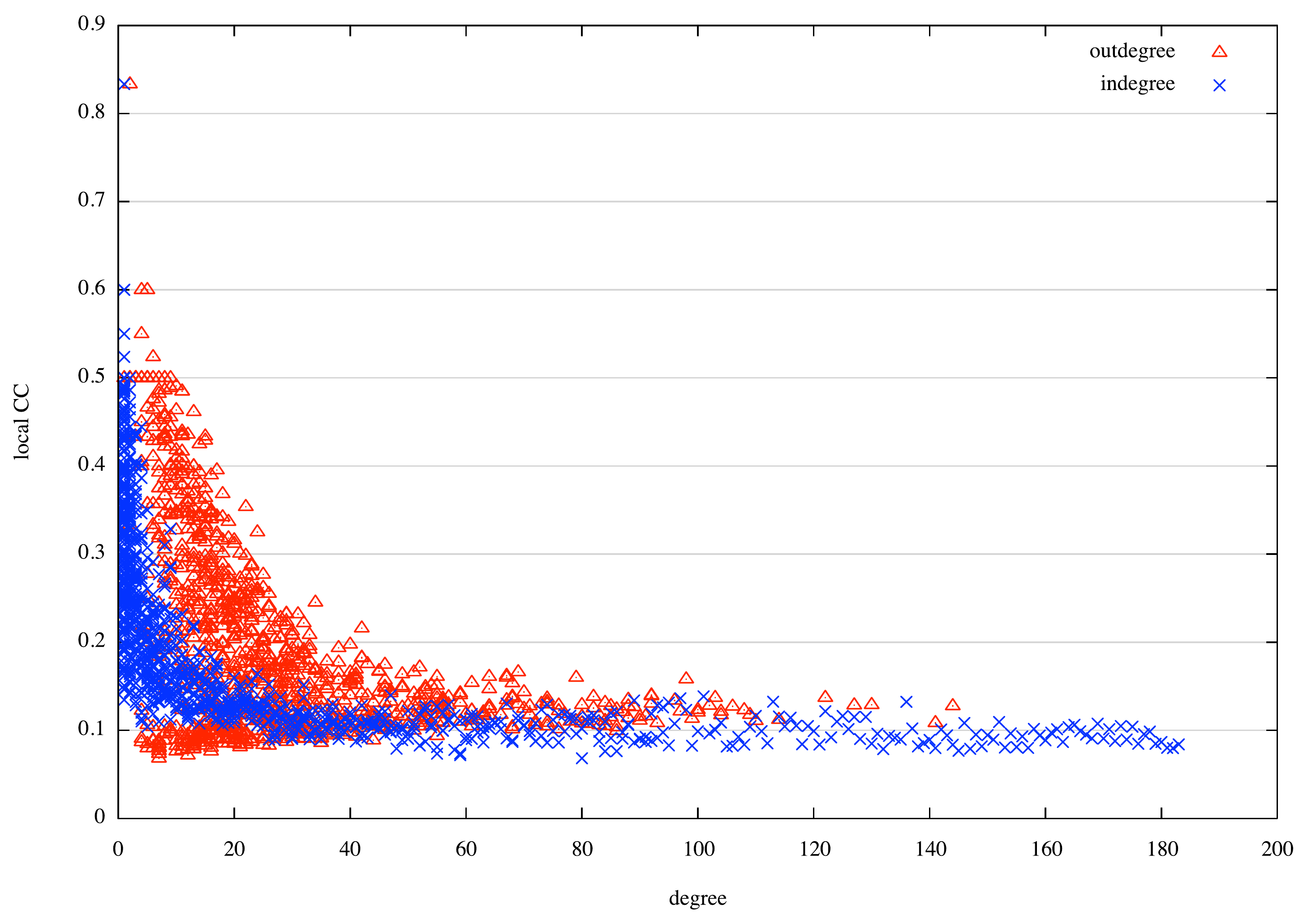}
					\label{fig4:b}
				}
				\caption{
				    In scale-free networks well connected nodes are less clustered.
				    Here we show for size $N = 1,024$ and in-degree scale-free nets
				    with the two levels of mean CC (low and high).
				    For the case of out-degree scale-free nets, degree distributions are switched
				    so that out-degree becomes in-degree and vice versa.
				}
				\label{fig4}
			\end{figure}             
            
            Then, we verify that better-connected nodes possess higher firing
            rate than any other type of node.
            In Fig.~\ref{fig5} we present how the two different degree distributions (in and 
            out) are related to spiking activity in scale-free networks.
            Fig.~\ref{fig5:a} shows this for low mean clustered scale-free networks, whereas
            Fig.~\ref{fig5:b} shows it for high mean clustered networks.
            For the case of out-degree scale-free networks, that is, networks that possess
            broadcasting hubs we do not observe any correlation between node out-degree and
            firing activity.
            This occurs in out-degree scale-free networks with low and high mean CC, and across all
            system sizes.
            However, for the case of in-degree scale-free networks, that is, networks that include
            absorbing hubs, we observe a positive correlation between node in-degree and spiking.
            This behaviour occurs in all system sizes.
            Interestingly, the behaviour of in-degree scale-free networks with low mean CC differ
            from the behaviour of the same type of network with high mean CC.
            Both exhibit a positive correlation between in-degree and total number of spikes, 
            nevertheless low mean clustered networks exhibit a linear trend, whereas high mean 
            clustered ones exhibit a non-linear trend.
            This suggests that as a network becomes more clustered (and more \emph{small-worldly})
            the activity of their nodes exhibit a more quadratic dependency of the node's in-degree.
            We do not explore this hypothesis in the present work, but it is the direction of future
            research.

			\begin{figure}[t]
				\centering
				\subfloat[Low mean CC]{%
					\includegraphics[scale=0.28]{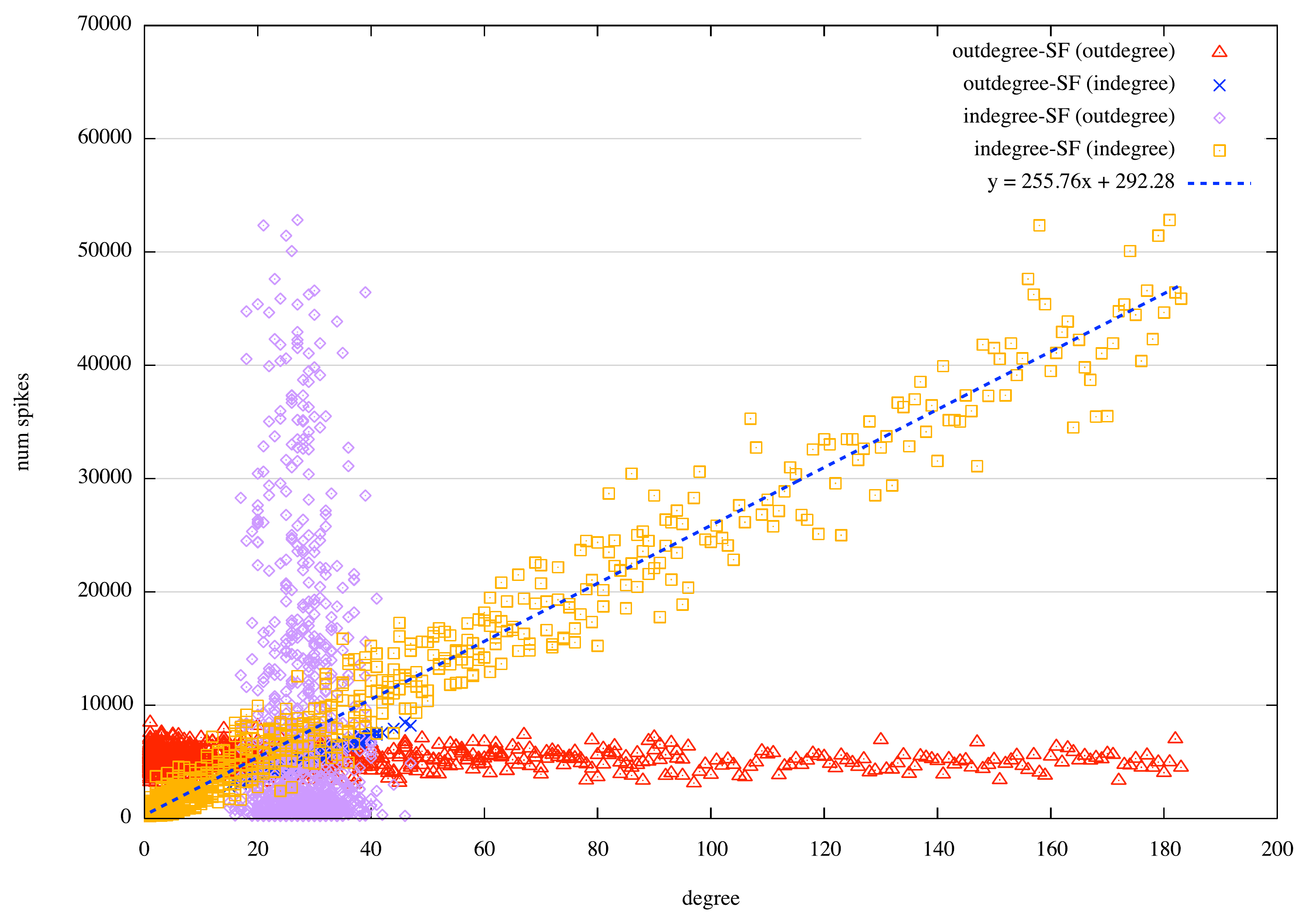}
					\label{fig5:a}
				}\\
				\subfloat[High mean CC]{%
					\includegraphics[scale=0.28]{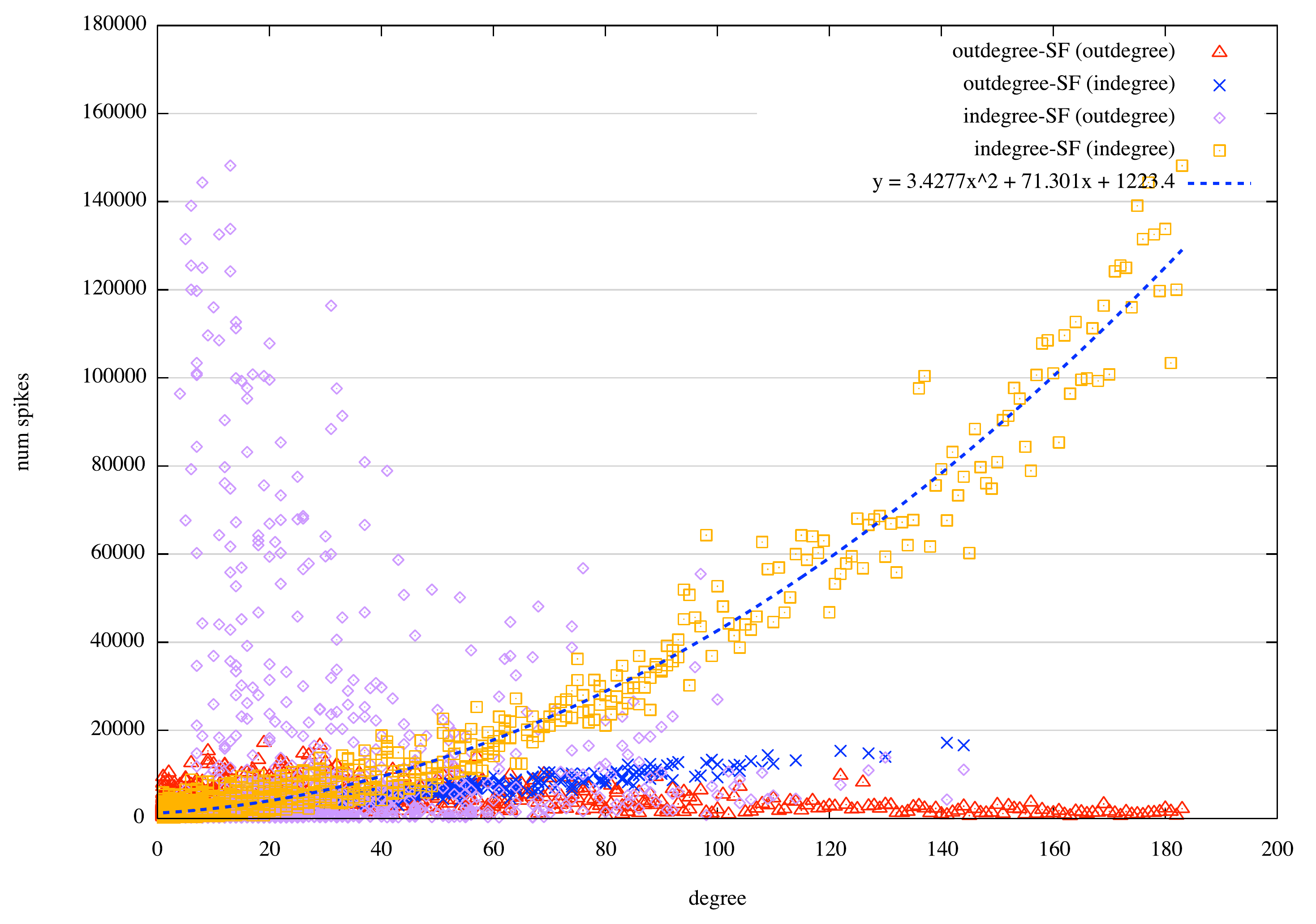}
					\label{fig5:b}
				}
				\caption{
				    In in-degree scale-free nets the absorbing hubs spike more than
				    any other type of node and their in-degree is correlated with the amount
				    of spikes fired.
				    For networks with higher mean CC this correlation is quadratic.
				    We show this behaviour for a network of size $N = 1,024$.
				}
				\label{fig5}
			\end{figure}

            At first, the observation of a positive correlation between spiking and node in-degree
            would seem very obvious, in the sense that nodes with high
            in-degree are driven beyond threshold frequently by the action of their in-neighbors.
            However a high in-degree cannot explain completely the high firing rate of this type
            of nodes, because this would also predict that nodes in fully-connected networks would
            possess a high firing rate, which is not the case 
            (see Fig.~\ref{fig6:a}).
            What is happening in fully-connected networks that prevents nodes from firing as much
            as the other heterogeneous structures considered even when these nodes are massively
            connected?
            We suggest that an obstructing behaviour is occurring in this globally-coupled structure.
            We name this phenomenon \emph{spike jamming} and we will discuss it in detail in
            Sect.~\ref{results:upperbound}.
            
            Random networks behave somehow
            similarly to scale-free networks in the sense that there is a positive correlation
            between node in-degree and spiking, but not between this latter and node out-degree.
            However, this type of structure does not reach the same amount of spiking per node as scale-free
            networks due to the random nature of their connectivity (not shown here).
            Thus, the presence of hubs account for the high firing rate in heterogeneous structures.
            
            Another question that arises at this point is the following.
            Does a high in-degree in scale-free networks account by itself for a high spiking
            activity? Or is it the joint action of in- and out-degree that explain this particular
            behaviour? In other words, could this high firing rate be explained by a specific
            configuration of in- and out-degree?
            To explore this question we considered the ratio of in-degree to out-degree per node,
            which for a node $i$ is given by:
            
            \begin{equation}
                \rho_{i} = \frac{\sum_{k = 1}^{N} A_{ki}}{\sum_{j = 1}^{N} A_{ij}}
            \end{equation}
            
            \noindent where the numerator is the in-degree of node $i$, and the denominator is its
            out-degree.
            The quantity $\rho$ is equal to unity when a node has the same number of incoming and
            outgoing connections. In fully-connected networks all nodes possess this property.
            If $\rho_{i} > 1$, then the in-degree of node $i$ is larger than its out-degree, and
            $\rho_{i} < 1$ when the opposite occurs.
            
            From fully-connected networks we have learned that homogeneity in node degree, that
            is, $\rho = 1$ for every node, is not a property suitable for spiking.
            Moreover, we verify this fact in heterogeneous structures where nodes with 
            $\rho\gg 1$ fire more than any other nodes.
            Fig.~\ref{fig7:a} shows this particular behaviour.
            There we present how spiking is improved as $\rho$ grows larger than unity.
            The solid black line at left of Fig.~\ref{fig7:a} marks the point where $\rho = 1$, 
            ie. where in-degree equals out-degree.
            For all heterogeneous structures considered (excluding random networks) a larger in-degree than out-degree
            correlates with higher firing rate.
            
 			\begin{figure}[t]
				\centering
				\subfloat[Total number of spikes]{%
					\includegraphics[scale=0.28]{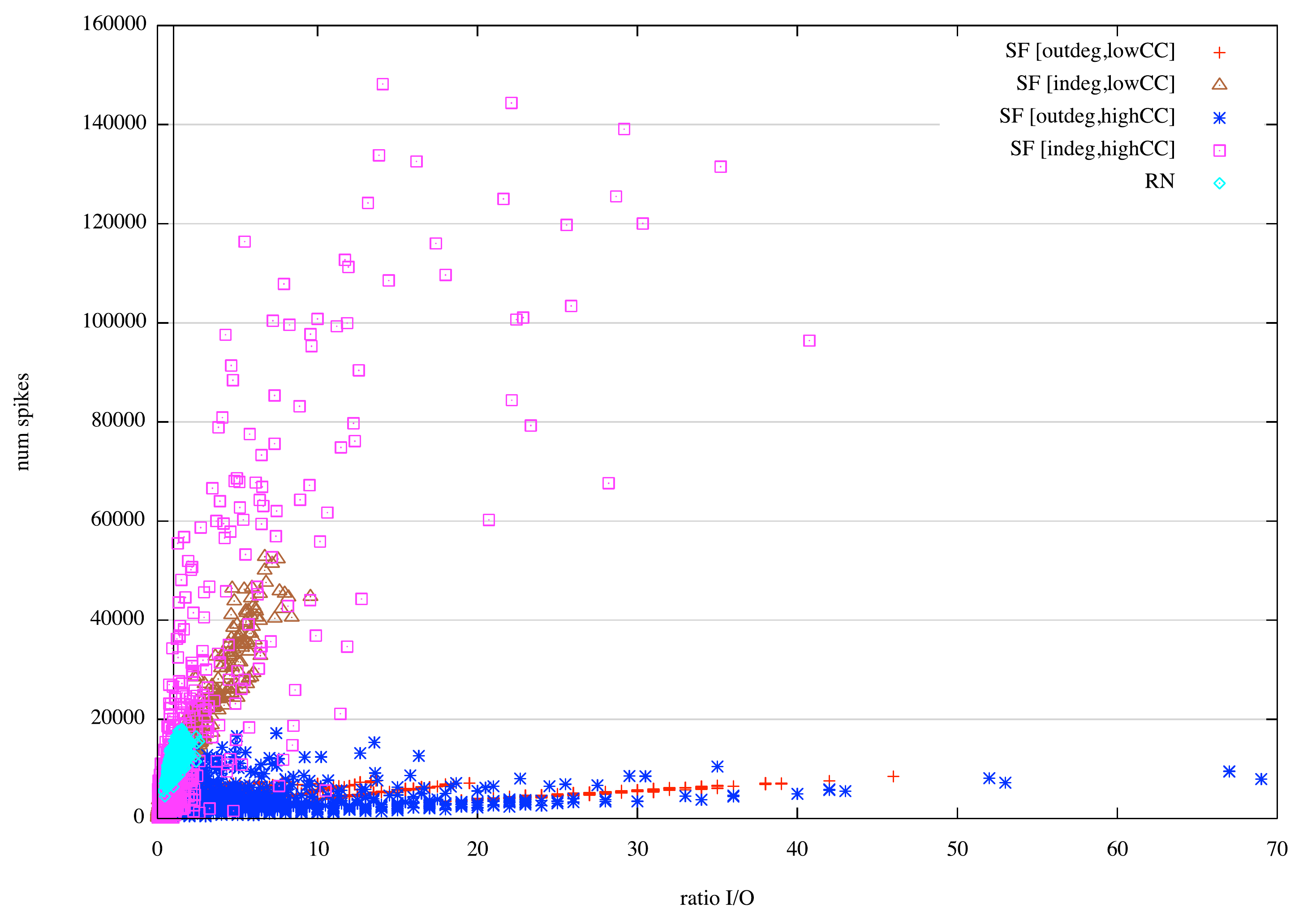}
					\label{fig7:a}
				}\\
				\subfloat[Total mean node success]{%
					\includegraphics[scale=0.28]{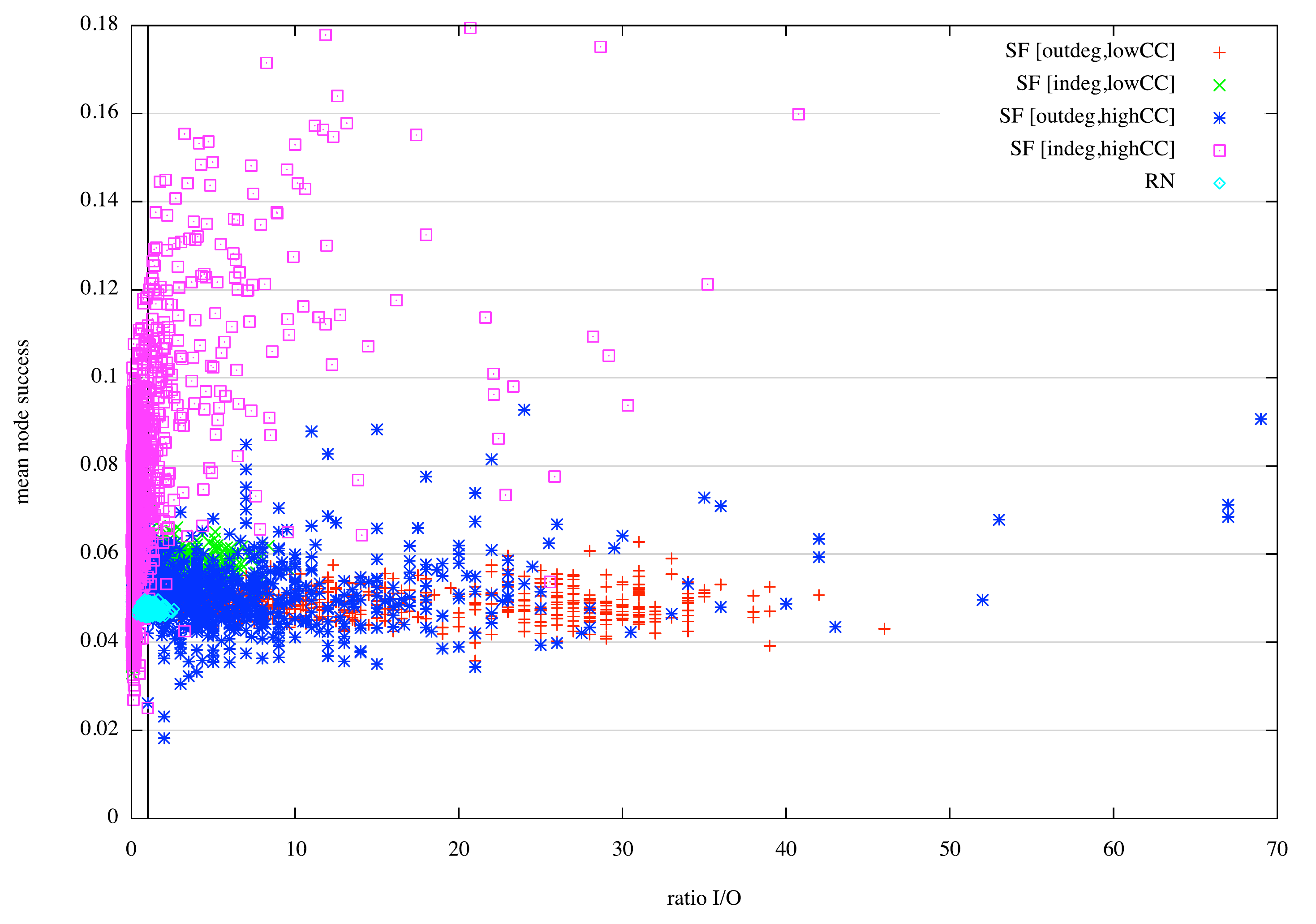}
					\label{fig7:b}
				}
				\caption{
				    Total number of spikes and their success per ratio in-degree/out-degree
				    for heterogeneous nets of size $N = 1,024$.
				    The solid line at the left lies at $\rho = 1$ where in-degree equals out-degree.
				    An equal number of incoming and outgoing connections cannot account for
				    higher spike rate and spike success.
				}
				\label{fig7}
			\end{figure}           
            
            In summary, in fully-connected networks nodes fire less than in any other topology.
            This is explained by the homogeneity of the nodes comprising the network, which 
            give rise to the phenomenon of spike jamming.
            In heterogeneous topologies, scale-free networks fire more than random networks,
            and the firing activity is improved by the presence of absorbing hubs and high mean
            CC, which implies a larger degree of small-world-ness.
            However, for this to happen absorbing hubs must possess the right amount of outgoing
            connections which is represented by $\rho\gg1$.
            
        \subsection{Scale-free topologies comprise more successful nodes}
        
            Spiking is not all that matters, since we should also consider the fate of a spike
            that has just been emitted.
            Here, we consider a \emph{successful spike} one that triggers subsequent spikes from the nodes
            in the out-neighborhood of the node where the initial spike originated.
            As we are interested in the propagation of activity within the system, we would like to 
            observe the sustained activation of nodes in subsequent time steps.
            This is where the notion of the \emph{branching ratio} comes to hand.
            The branching ratio $\sigma$ is defined as the ratio of descendants that
            become active at time $t + 1$ to ancestors that were active at time $t$.
            The quantity $\sigma$ has been used to characterize the critical state of a
            system~\cite{beggs2003neuronal} and to identify the regimes surrounding such a state.
            When $\sigma < 1$ the system is subcritical and activity dies out quickly,
            when this value is above unity, the system is supercritical and activity gets
            amplified pathologically at each time step.
            In between these two states lies the critical state in which activity is sustained
            until finite-size effects take place, during this regime $\sigma$ is equal to unity
            for a prolonged period of time.
            
            Recall that in Sect.~\ref{model:nodeSuccess} we defined the success of any give node $i$ as the
            fraction of out-neighbors that become active at time $t+1$ when node $i$ spiked at
            time $t$.
            This quantity is similar to the branching ratio $\sigma$ in the sense that it estimates
            the amount of activity sustained in subsequent time steps.
            However, unlike $\sigma$ our measure of node success is a \emph{local} estimation
            of performance, which has more natural implications in the context of neuronal
            networks, in which a neuron does not have access to global metrics regarding
            the structure of the network.
            
            As mentioned above, spiking does not imply success, and node success as defined above
            is intimately related to the notion of criticality.
            Here we repeat the same questions that we considered in the previous section,
            namely, how successful are nodes with high local CC? how successful are hubs? and
            finally, how network structure affect node success?
            
            For scale-free networks, nodes with low local CC are the most successful nodes.
            This behaviour is more evident in in-degree scale-free networks with high mean CC
            (see Fig.~\ref{fig3:b}).
            These low local CC nodes are the hubs, however unlike the firing activity, mean node
            success per node does not exhibit a very clear correlation between in- or out-degree
            and success (see Fig.~\ref{fig7:b}).
            Finally, random networks do not exhibit any particular pattern regarding the success
            of their nodes.
                
            Which is the most successful topology? In other words, what is the structure that
            maximises the success per node?
            To answer this question we estimated the mean node success of the system per time step
            (see Sect.~\ref{model:nodeSuccess})
            for all the topologies considered.
            
            For all system sizes we observe that fully-connected networks are the type of structure
            that performs worst (see Fig.~\ref{fig6:b}), followed by random networks.
            For the case of scale-free networks, in-degree scale-free networks with high mean CC
            are the most successful topologies.
            Second in place are in-degree scale-free networks with low mean CC, followed by
            out-degree scale-free networks with high and low mean CC, respectively.
            Thus, absorbing hubs in a scale-free structure allow for more node success per node
            if accompanied by a high degree of small-world-ness.
            
            Recall that the scale-free and random topologies have the same number of edges.
            Therefore, in-degree scale-free-ness with high mean CC is the permutation of edges
            that maximises node success.
            Fig.~\ref{fig8} shows the aforementioned behaviour for networks of size $N = 1,024$,
            the same is observed in all the other system sizes considered.
            
 			\begin{figure}[t]
				\centering
				\includegraphics[scale=0.3]{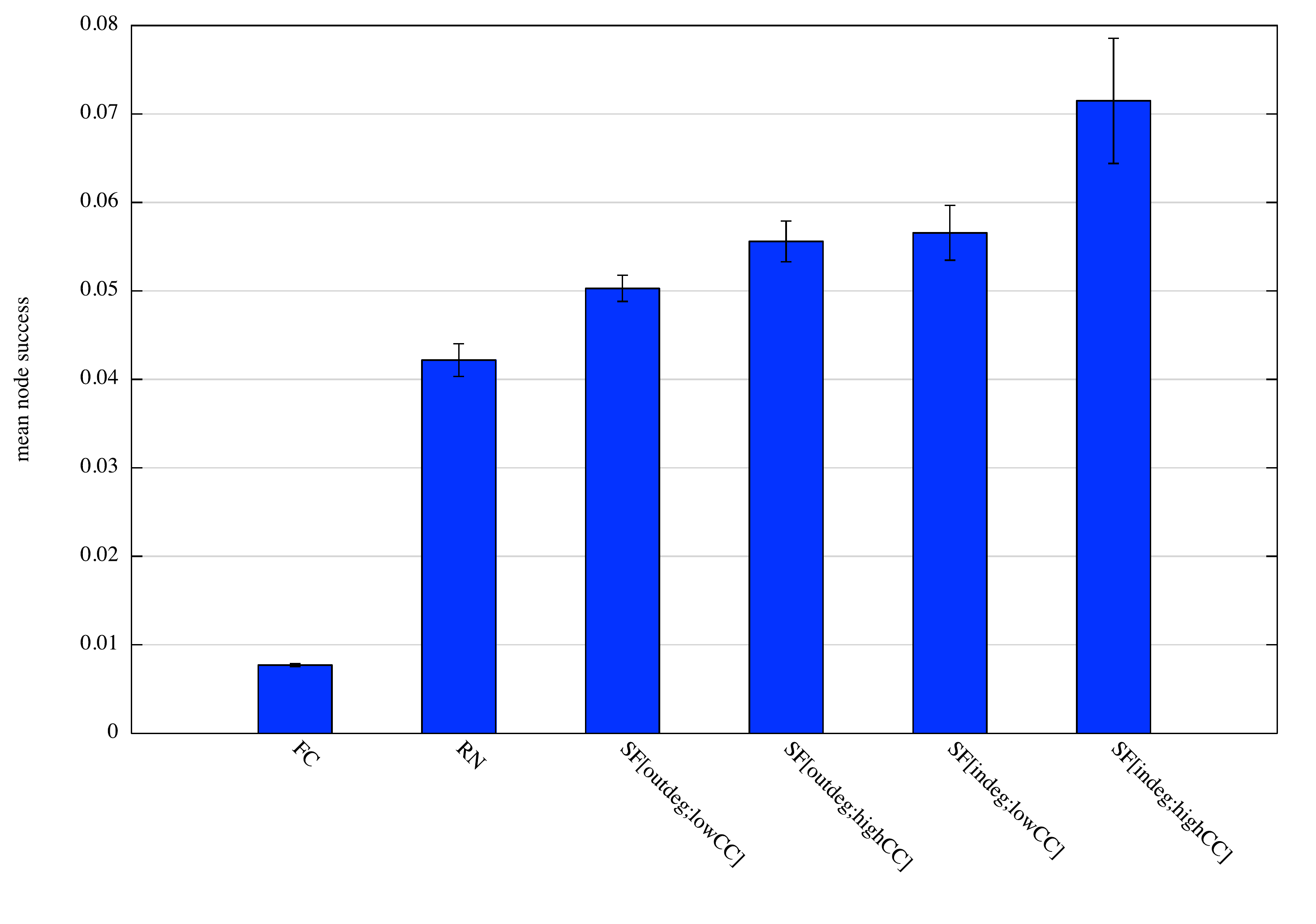}
				\caption{
				    Mean node success for scale-free
				    and random networks of size $N = 1,024$.
				    In-degree scale-free networks with high mean CC possess more successful nodes than any other
				    topology considered.
				    (Error bars denote standard deviations.)
				}
				\label{fig8}
			\end{figure}             
            
            Moreover, the value of the mean node success per time step is upper bounded.
            For a system to remain in the critical regime the mean node success must remain
            below a certain value.
            In other words, a high mean node success is related to the supercritical
            regime in which nodes fire constantly (not shown here);
            so that, if we were interested in maximizing the value of the mean node success we would have
            to leave the critical regime.
            Additionally, we report that the upper bound of the total mean node success at criticality decreases as the 
            system size grows.
            An example of this is shown in Fig.~\ref{fig6:b} where the maximum value of node success for
            fully-connected topologies decrease with system size.
            Although not shown here, this phenomenon is observed in all other topologies and system sizes.
        
        \subsection{Upper bound of mean node success for fully-connected nets}
	        \label{results:upperbound}
            As mentioned in the previous section, fully-connected networks perform worst as
            measured by the mean node success per time step.
            In this section we derive an analytical
            expression for the upper bound of this metric for globally-coupled structures.
    
            Recall that $h_i(t)$ denotes the membrane potential of node $i$ at time step $t$,
            and that $\theta>0$ denotes the threshold 
            membrane potential required to trigger a spike in a node: node $i$ will spike at $t$ when $h_i(t)\geq\theta$.             
            At $t=0$ the membrane potentials take values $h_i(0)<\theta$ for all $i$. 
            These potentials are then \emph{driven externally} each time step until the membrane potential of one 
            node is taken above the threshold membrane potential - triggering an avalanche.
            During an avalanche the membrane potentials evolve as follows:
            
            \begin{equation}\label{h_dynamics}
                h_i(t+1)=
                \begin{cases}
                    0 & \text{if $h_i(t)\geq\theta$} \\
                    h_i(t)+\displaystyle\sum_{j\in\mathcal{I}_i}w_{ij}s_j(t) & \text{otherwise}, \\
                \end{cases}
            \end{equation}
            where $\mathcal{I}_i$ denotes the set of in-neighbors of $i$, and recall that $w_{ij}$ describes the coupling
            between nodes $i$ and $j$, and that $s_j(t)=1$ if node $j$ spikes at $t$ and 0 otherwise, i.e.
            \begin{equation}
                s_j(t)=
                \begin{cases}
                    1 & \text{if $h_j(t)\geq\theta$} \\
                    0 & \text{otherwise}. \\
                \end{cases}
            \end{equation}
            Note that the membrane potentials are \emph{not} driven during the avalanche. The avalanche ends when there 
            are no more spiking nodes.
            Afterwards, the membrane potentials are then driven again each time step until another avalanche is    
            triggered. This process is repeated until the simulation ends.

            In a fully connected network each node has $(N-1)$ out-neighbors, where $N$ denotes the number of nodes in 
            the network. Specifically, the out-neighbors of node $i$ are all the nodes $j$ such that $j\neq i$. 
            Recall that the node success $\varphi_i(t)$ of node $i$ at time step $t$ is defined as the 
            fraction of out-neighbors of $i$ which spike at time step $t+1$, given that $i$ spikes at $t$.
            Therefore, if $i$ spikes at $t$\footnote{$\varphi_i(t)$ is undefined if node $i$ does not spike at $t$.}, 
            Eqn.~\eqref{eq:phi} becomes
            
            \begin{equation}\label{phi_t}
                \varphi_i(t)=\frac{S(t+1)}{N-1}
            \end{equation}
            in a fully connected network, where $S(t)$ denotes the number of nodes which spike at time step $t$. Note that 
            the right-hand side of the above expression
            is independent of $i$. This reflects the fact that $\varphi_i(t)$ is identical for all nodes $i$ which spike 
            at 
            $t$. For this reason we will henceforth
            omit the subscript $i$, and deal only with the quantity $\varphi(t)$, the node success of \emph{any} node 
            which 
            spikes at $t$.
            This observation is only valid for fully-connected networks in which all nodes receive connections from
            each other and send connections likewise.

            Consider the mean node success for nodes which spike during any period of $\tau$ time steps. 
            Since $S(t)$ nodes
            spike at $t$, this is given by
            
            \begin{equation}\label{mean_phi_def}
                \langle\varphi\rangle=\frac{1}{\mathcal{S}}\sum_{t=1}^{\tau}S(t)\varphi(t),
            \end{equation}
            where without loss of generality we have chosen the $\tau$ timsteps to be $t=1,2,\dotsc,\tau$, and
            \begin{equation}\label{calS_def}
                \mathcal{S}\equiv\sum_{t=1}^{\tau}S(t)
            \end{equation}
            
            is the total number of spikes which occur during this period.
            A crucial aspect of the dynamics described above is that $h_i(t+1)=0$ if node $i$ spikes at 
            time $t$. It is therefore impossible for 
            node $i$ to spike on two adjacent time steps -- after spiking, the membrane potential of node $i$ is 
            \emph{frozen} (due to refractoriness) to be zero for a single time step, 
            during which time it cannot `accumulate activity' from spiking in-neighbors in the manner described by 
            Eqn.~\eqref{h_dynamics}. 
            This constraint can be expressed mathematically as
            \begin{equation}\label{freeze_constraint}
                s_i(t)+s_i(t+1)\leq 1 \quad\text{for all $i$, $t$},
            \end{equation}
            and leads to the following theorem:
            \begin{theorem}\label{theorem}
                $\langle\varphi\rangle$ has an upper bound of
                \begin{equation}\label{phi_max}
                    \langle\varphi\rangle_{\text{max}}=\frac{N}{2(N-1)}
                \end{equation}
                in the limit $\tau\to\infty$. This upper bound is realised when $S(t)=N/2$ for all $t>1$.
            \end{theorem}
            
            \begin{proof}
                Substituting Eqn.~\eqref{phi_t} into Eqn.~\eqref{mean_phi_def} gives
                \begin{equation}\label{proof_phi}
                    \langle\varphi\rangle=\frac{1}{\mathcal{S}(N-1)}\sum_{t=1}^{\tau}S(t)S(t+1).
                \end{equation}
                Defining the quantity
                \begin{equation}\label{tilde_s_def}
                    \tilde{S}(t)\equiv S(t)/\sqrt{\mathcal{S}},
                \end{equation}
                Eqn. \eqref{proof_phi} can be expressed as
                \begin{equation}\label{proof_phi_2}
                    \langle\varphi\rangle=\frac{1}{(N-1)}\sum_{t=1}^{\tau}\tilde{S}(t)\tilde{S}(t+1).
                \end{equation}
                Now, taking the summation of Eqn. \eqref{freeze_constraint} over all nodes $i$ yields
                \footnote{In obtaining Eqn. \eqref{st_constraint} we use the fact that if $x\leq a$ and $y\leq b$ then 
                $x+y\leq a+b$.}
                \begin{equation}\label{st_constraint}
                    S(t)+S(t+1)\leq N \quad\text{for all $t$}
                \end{equation}
                after noting that
                \begin{equation}
                    S(t)=\sum_is_i(t).
                \end{equation}
                Subtracting $S(t+1)$ from both sides of Eqn. \eqref{st_constraint} and then multiplying throughout by 
                $S(t+1)$ yields
                \begin{equation}
                    S(t)S(t+1)\leq S(t+1)\bigl[N-S(t+1)\bigr]\quad\text{for all $t$},
                \end{equation}
                which can be expressed as
                \begin{equation}
                    \tilde{S}(t)\tilde{S}(t+1)\leq \tilde{S}(t+1)\Biggl[\frac{N}{\sqrt{\mathcal{S}}}-\tilde{S}(t+1)\Biggr]       \quad\text{for all $t$}
                \end{equation}
                after dividing both sides by $\mathcal{S}$.
                \footnote{In obtaining these equations we have used the following rules for manipulating inequalities: if 
                $x\leq a$ then $x+c\leq a+c$ for all 
                $c$; if $x\leq a$ then $xc\leq ac$ if $c\geq 0$; if $x\leq a$ then $x/c\leq a/c$ if $c>0$. In applying the
                second of these rules 
                we used $c=s(t+1)$, which is $\geq 0$. In applying the last of these rules we used $c=\sqrt{\mathcal{S}}$,
                which is $>0$.}
                The right-hand side of the above inequality, which we denote as
                \begin{equation}\label{alpha_def}
                    \xi(t+1)\equiv\tilde{S}(t+1)\Biggl[\frac{N}{\sqrt{\mathcal{S}}}-\tilde{S}(t+1)\Biggr],
                \end{equation}
                is hence an upper bound for term $t$ on the right-hand side of Eqn. \eqref{proof_phi_2}. Therefore we can 
                write
                \begin{equation}\label{proof_phi_3}
                    \langle\varphi\rangle\leq\frac{1}{(N-1)}\sum_{t=1}^{\tau}\xi(t+1).
                \end{equation}
                Now, $\xi(t+1)$ is maximised when $\tilde{S}(t+1)=N\big/\bigl(2\sqrt{\mathcal{S}}\bigr)$. Therefore 
                $\tilde{S}(t+1)=N\big/\bigl(2\sqrt{\mathcal{S}}\bigr)$ for $t=1,2,\dotsc,\tau$ provides an upper bound for
                $\langle\varphi\rangle$, which,
                from substituting the aforementioned $\tilde{S}(t+1)$ into into Eqns.~\eqref{alpha_def} and 
                \eqref{proof_phi_3}, can be shown to be
                \begin{equation}\label{proof_phi_max}
                    \langle\varphi\rangle_{\text{max}}=\frac{1}{(N-1)}\biggl(\frac{N}{2}\biggr)^2\frac{\tau}{\mathcal{S}}.
                \end{equation}
                However, from Eqn.~\eqref{tilde_s_def}, $\tilde{S}(t+1)=N\big/\bigl(2\sqrt{\mathcal{S}}\bigr)$ for        
                $t=1,2,\dotsc,\tau$ corresponds to $S(t+1)=N/2$
                for $t=1,2,\dotsc,\tau$, and hence from Eqn.~\eqref{calS_def} also corresponds to
                \begin{equation}
                    \mathcal{S}=S(1)+\sum_{t=2}^{\tau}S(t)=S(1)+(\tau-1)N/2.
                \end{equation}
                Substituting this into Eqn.~\eqref{proof_phi_max} gives
                \begin{equation}
                    \langle\varphi\rangle_{\text{max}}=\frac{1}{(N-1)}\biggl(\frac{N}{2}\biggr)^2\frac{\tau}{S(1)+(\tau-1)N/2}.
                \end{equation}
                We emphasise that this is realised when $S(t+1)=N/2$ for $t=1,2,\dotsc,\tau$, or equivalently, when 
                $S(t)=N/2$ for $t=2,3,\dotsc,\tau+1$.
                Noting that $S(1)$ cannot exceed $N$, Theorem \ref{theorem} results when the limit $\tau\to\infty$ is 
                taken.
            \end{proof}
            
            Some remarks are due with regards to Theorem \ref{theorem}. 
            Firstly, $\langle\varphi\rangle$ as $\tau\to\infty$, which we henceforth refer to simply as $\langle\varphi\rangle$,
            describes the mean node success over all nodes over all time.
            
            Secondly, Theorem \ref{theorem} applies in a very general way to fully-connected networks, in the sense
            that for the purpose of proving the 
            theorem we have made no assumptions regarding how the system is driven between avalanches, or the initial 
            values of
            the membrane potentials. 
            Furthermore, we have made no assumptions regarding the specific values of $\theta$ or $w_{ij}$.
            
            Thirdly, since Theorem \ref{theorem} pertains to a fully-connected network, it follows that any network with
            $\langle\varphi\rangle>\langle\varphi\rangle_{\text{max}}$ \emph{cannot} be a fully connected network. In a similar
            vein, if one wishes to construct a network with $\langle\varphi\rangle>\langle\varphi\rangle_{\text{max}}$ 
            starting from a fully connected network, it is necessary that some connections between nodes are removed,
            that is, such a network should part from a massively connected to a less connected structure. 
            
            Fourthly, 
            $\langle\varphi\rangle_{\text{max}}$ decreases monotonically with $N$, and $\langle\varphi\rangle_{\text{max}}\to 1/2$ 
            in the thermodynamic limit, i.e., $N\to\infty$. 
            
            Finally, and most importantly, the theorem is non-trivial in the sense that one can easily conceive of networks of size
            $N$ whose global node successes can potentially exceed $\langle\varphi\rangle_{\text{max}}$. For instance, consider 
            the network corresponding to a `directed ring', where $A_{12}=1,A_{23}=1,\dotsc,A_{(N-1)N}=1,A_{N1}=1$, and $A_{ij}=0$ 
            for all other elements of the adjacency matrix. If $w_{ij}=\theta$, then assuming, without loss of generality, that first 
            spike in the network occurs on time step $t=1$ at node 1, then the spike propagates around the ring indefinitely: at 
            $t=2$, node 2 spikes; at $t=3$ node 3 spikes, at $t=N$ node $N$ spikes, at $t=N+1$ node 1 spikes, etc. In this case it 
            is easy to see that $\langle\varphi\rangle=1$, which is greater than $\langle\varphi\rangle_{\text{max}}$ for $N>2$. 
            Therefore the existence of an upper bound for fully-connected networks stems from some particular property of 
            their topology in combination with the dynamics described in Sect.~\ref{model}. 
            
            What is this property of fully-connected networks which places this upper bound on their node success? 
            As alluded to earlier, it is the fact that nodes are frozen for the time step after they spike which gives 
            rise to 
            the upper bound in fully-connected networks. 
            This behaviour gives rise to the phenomenon of \emph{spike jamming} that we mentioned in 
            Sect.~\ref{results:spiking}, and which we describe in detail below.
            
            Consider a single node $i$ firing at time step $t$ in a fully-connected
            network. For this node to be maximally successful, it must trigger all $N-1$ of its out-neighbors, i.e., all other
            nodes in the network, to spike at $t+1$. Suppose this happens, in which case $\varphi_i(t)=1$. Consider now one
            of the nodes $j\neq i$ which spikes at $t+1$. For $j$ to be maximally successful, all other nodes in the network
            must spike at $t+2$. However, on account of refractoriness, this is impossible. To elaborate, at $t+1$, all nodes except
            for $i$, and including $j$, are spiking. Therefore all these nodes must be frozen at $t+2$ - they cannot spike at
            $t+2$. On the other hand $i$, which spiked at $t$, while frozen at $t+1$, is free to spike at $t+2$. Hence, at best,
            only one of the $N-1$ out-neighbors of $j$, namely $i$, can spike at $t+2$, and therefore at best 
            $\varphi_j(t+1)=1/(N-1)$. For large $N$, $j$ is clearly very unsuccessful. The same applies for all other nodes which
            fire at $t+1$. Hence the result is that, while $i$ is maximally successful, the remaining $N-1$ nodes are extremely
            unsuccessful, and hence on average the whole network is unsuccessful during this avalanche - which we assume ends at
            $t+2$. This example illustrates the effect which underpins the upper bound for fully-connected networks:
            if a node $i$ spikes synchronously with one of its out-neighbors, then that out-neighbor is frozen for the next 
            time step, and hence cannot spike on the time step after $i$ spikes, which curtails the potential node success of $i$, and correspondingly the propagation of spikes throughout the network. Hence we refer to this effect as spike jamming. 
            Note that aforementioned effect occurs in all networks, not just fully-connected networks.
            However, fully-connected networks are special in that all nodes are out-neighbors of each other, and hence
            this effect has more potential to curtail the node success in fully-connected networks than in any other network.

\section{Discussion}
	\label{sect:discussion}
    In this paper we have presented arguments regarding the poor performance 
    (in terms of spiking of individual nodes and their success) 
    of fully-connected networks
    at criticality showing at the same time that scale-free networks perform much better than any other
    topology when paired with the small-world property.
    
    Given that the heterogeneous topologies possess exactly the same number of edges, we
    conclude that scale-free-ness with high degree of small-world-ness is a permutation of edges that
    allows nodes to be more successful and to be more active in terms of the number of spikes emitted.
    In particular, we have verified the statement above for the case of in-degree scale-free networks, which
    feature the presence of \emph{absorbing hubs}.
    However, real-world networks often comprise a more complex ecosystem in which absorbing hubs and
    broadcasting hubs coexist in the same network adding another layer of complexity to the dynamics within
    the system.
    Moreover, it is often the case that the structure of real-world networks is not static, but they possess
    mechanisms by which nodes become connected and disconnected over time as well as network growth or shrinkage;
    features that affect the collective activity in ways that cannot be predicted with the current model.
    
    This leads us to the next consideration. What real system are we describing with the current model?
    From a certain point of view, the model used here is very limited or simplistic, 
    however a model of integrate-and-fire units
    can actually be a simplified model of many phenomena in nature.
    A model of threshold units that accumulate activity from their vicinity and then propagate it when going
    beyond threshold can be used in principle to model the spread of epidemics, piles of granular matter,
    the release of energy and relaxation of tectonic plates, the effects of a stock market crash, 
    and the activity of neurons of the brain, among others.
    Thus, we believe that our model has a broad range of applications in diverse contexts, in which the presence
    of particular network properties such as the small-world property and long-tailed degree distributions
    have immediate effects on the dynamics of the system, be it the spread of a disease in a population, the propagation 
    of stimuli on cortical networks, or the spread of rumors and fads
    within a social network.
    Moreover, we believe that the introduction of the analysis of the success of a spike can be applied to the
    situations mentioned above.
    In other contexts, a spike could be thought of the transmission of an infection among contacts, 
    the death of a species in models of ecosystems, the failure of a power generator in power networks, and even in 
    on-line social networks such as Facebook or Twitter we might regard a spike as the action of writing a
    \emph{post} or a \emph{tweet}.
    In all these contexts, the fate of a spike is as relevant to the collective dynamics as is the network topology.
    Here we have shown that the combination of individual dynamics of nodes and topology determine the success of the 
    spikes that spread across the system.

\begin{acknowledgments}
    VHU would like to thank the Mexican National Council on Science and 
	Technology (CONACYT) fellowship no. 214055 for partially funding this work.
\end{acknowledgments}

	\bibliography{biblio}

\begin{thebibliography}{25}%
\makeatletter
\providecommand \@ifxundefined [1]{%
 \@ifx{#1\undefined}
}%
\providecommand \@ifnum [1]{%
 \ifnum #1\expandafter \@firstoftwo
 \else \expandafter \@secondoftwo
 \fi
}%
\providecommand \@ifx [1]{%
 \ifx #1\expandafter \@firstoftwo
 \else \expandafter \@secondoftwo
 \fi
}%
\providecommand \natexlab [1]{#1}%
\providecommand \enquote  [1]{``#1''}%
\providecommand \bibnamefont  [1]{#1}%
\providecommand \bibfnamefont [1]{#1}%
\providecommand \citenamefont [1]{#1}%
\providecommand \href@noop [0]{\@secondoftwo}%
\providecommand \href [0]{\begingroup \@sanitize@url \@href}%
\providecommand \@href[1]{\@@startlink{#1}\@@href}%
\providecommand \@@href[1]{\endgroup#1\@@endlink}%
\providecommand \@sanitize@url [0]{\catcode `\\12\catcode `\$12\catcode
  `\&12\catcode `\#12\catcode `\^12\catcode `\_12\catcode `\%12\relax}%
\providecommand \@@startlink[1]{}%
\providecommand \@@endlink[0]{}%
\providecommand \url  [0]{\begingroup\@sanitize@url \@url }%
\providecommand \@url [1]{\endgroup\@href {#1}{\urlprefix }}%
\providecommand \urlprefix  [0]{URL }%
\providecommand \Eprint [0]{\href }%
\providecommand \doibase [0]{http://dx.doi.org/}%
\providecommand \selectlanguage [0]{\@gobble}%
\providecommand \bibinfo  [0]{\@secondoftwo}%
\providecommand \bibfield  [0]{\@secondoftwo}%
\providecommand \translation [1]{[#1]}%
\providecommand \BibitemOpen [0]{}%
\providecommand \bibitemStop [0]{}%
\providecommand \bibitemNoStop [0]{.\EOS\space}%
\providecommand \EOS [0]{\spacefactor3000\relax}%
\providecommand \BibitemShut  [1]{\csname bibitem#1\endcsname}%
\let\auto@bib@innerbib\@empty
\bibitem [{\citenamefont {Newman}(2010)}]{newman2010networks}%
  \BibitemOpen
  \bibfield  {author} {\bibinfo {author} {\bibfnamefont {M.}~\bibnamefont
  {Newman}},\ }\href@noop {} {\emph {\bibinfo {title} {Networks: an
  introduction}}}\ (\bibinfo  {publisher} {Oxford University Press, Inc.},\
  \bibinfo {year} {2010})\BibitemShut {NoStop}%
\bibitem [{\citenamefont {Watts}\ and\ \citenamefont
  {Strogatz}(1998)}]{watts1998collective}%
  \BibitemOpen
  \bibfield  {author} {\bibinfo {author} {\bibfnamefont {D.}~\bibnamefont
  {Watts}}\ and\ \bibinfo {author} {\bibfnamefont {S.}~\bibnamefont
  {Strogatz}},\ }\href@noop {} {\bibfield  {journal} {\bibinfo  {journal}
  {Nature}\ }\textbf {\bibinfo {volume} {393}},\ \bibinfo {pages} {440}
  (\bibinfo {year} {1998})}\BibitemShut {NoStop}%
\bibitem [{\citenamefont {Humphries}\ and\ \citenamefont
  {Gurney}(2008)}]{humphries2008network}%
  \BibitemOpen
  \bibfield  {author} {\bibinfo {author} {\bibfnamefont {M.~D.}\ \bibnamefont
  {Humphries}}\ and\ \bibinfo {author} {\bibfnamefont {K.}~\bibnamefont
  {Gurney}},\ }\href@noop {} {\bibfield  {journal} {\bibinfo  {journal} {PLoS
  One}\ }\textbf {\bibinfo {volume} {3}},\ \bibinfo {pages} {e0002051}
  (\bibinfo {year} {2008})}\BibitemShut {NoStop}%
\bibitem [{\citenamefont {Barab{\'a}si}\ and\ \citenamefont
  {Albert}(1999)}]{barabasi1999emergence}%
  \BibitemOpen
  \bibfield  {author} {\bibinfo {author} {\bibfnamefont {A.-L.}\ \bibnamefont
  {Barab{\'a}si}}\ and\ \bibinfo {author} {\bibfnamefont {R.}~\bibnamefont
  {Albert}},\ }\href@noop {} {\bibfield  {journal} {\bibinfo  {journal}
  {Science}\ }\textbf {\bibinfo {volume} {286}},\ \bibinfo {pages} {509}
  (\bibinfo {year} {1999})}\BibitemShut {NoStop}%
\bibitem [{\citenamefont {Bak}\ \emph {et~al.}(1988)\citenamefont {Bak},
  \citenamefont {Tang},\ and\ \citenamefont {Wiesenfeld}}]{bak1988self}%
  \BibitemOpen
  \bibfield  {author} {\bibinfo {author} {\bibfnamefont {P.}~\bibnamefont
  {Bak}}, \bibinfo {author} {\bibfnamefont {C.}~\bibnamefont {Tang}}, \ and\
  \bibinfo {author} {\bibfnamefont {K.}~\bibnamefont {Wiesenfeld}},\
  }\href@noop {} {\bibfield  {journal} {\bibinfo  {journal} {Physical Review
  A}\ }\textbf {\bibinfo {volume} {38}},\ \bibinfo {pages} {364} (\bibinfo
  {year} {1988})}\BibitemShut {NoStop}%
\bibitem [{\citenamefont {Gutenberg}\ and\ \citenamefont
  {Richter}(1956)}]{gutenberg1956magnitude}%
  \BibitemOpen
  \bibfield  {author} {\bibinfo {author} {\bibfnamefont {B.}~\bibnamefont
  {Gutenberg}}\ and\ \bibinfo {author} {\bibfnamefont {C.~F.}\ \bibnamefont
  {Richter}},\ }\href@noop {} {\bibfield  {journal} {\bibinfo  {journal}
  {Annals of Geophysics}\ }\textbf {\bibinfo {volume} {9}},\ \bibinfo {pages}
  {1} (\bibinfo {year} {1956})}\BibitemShut {NoStop}%
\bibitem [{\citenamefont {Frette}\ \emph {et~al.}(1996)\citenamefont {Frette},
  \citenamefont {Christensen}, \citenamefont {Malthe-S{\o}renssen},
  \citenamefont {Feder}, \citenamefont {J{\o}ssang},\ and\ \citenamefont
  {Meakin}}]{frette1996avalanche}%
  \BibitemOpen
  \bibfield  {author} {\bibinfo {author} {\bibfnamefont {V.}~\bibnamefont
  {Frette}}, \bibinfo {author} {\bibfnamefont {K.}~\bibnamefont {Christensen}},
  \bibinfo {author} {\bibfnamefont {A.}~\bibnamefont {Malthe-S{\o}renssen}},
  \bibinfo {author} {\bibfnamefont {J.}~\bibnamefont {Feder}}, \bibinfo
  {author} {\bibfnamefont {T.}~\bibnamefont {J{\o}ssang}}, \ and\ \bibinfo
  {author} {\bibfnamefont {P.}~\bibnamefont {Meakin}},\ }\href@noop {}
  {\bibfield  {journal} {\bibinfo  {journal} {Nature}\ }\textbf {\bibinfo
  {volume} {379}},\ \bibinfo {pages} {49} (\bibinfo {year} {1996})}\BibitemShut
  {NoStop}%
\bibitem [{\citenamefont {Bak}\ \emph {et~al.}(1990)\citenamefont {Bak},
  \citenamefont {Chen},\ and\ \citenamefont {Tang}}]{bak1990forest}%
  \BibitemOpen
  \bibfield  {author} {\bibinfo {author} {\bibfnamefont {P.}~\bibnamefont
  {Bak}}, \bibinfo {author} {\bibfnamefont {K.}~\bibnamefont {Chen}}, \ and\
  \bibinfo {author} {\bibfnamefont {C.}~\bibnamefont {Tang}},\ }\href@noop {}
  {\bibfield  {journal} {\bibinfo  {journal} {Physics Letters A}\ }\textbf
  {\bibinfo {volume} {147}},\ \bibinfo {pages} {297} (\bibinfo {year}
  {1990})}\BibitemShut {NoStop}%
\bibitem [{\citenamefont {Eurich}\ \emph {et~al.}(2002)\citenamefont {Eurich},
  \citenamefont {Herrmann},\ and\ \citenamefont {Ernst}}]{eurich2002finite}%
  \BibitemOpen
  \bibfield  {author} {\bibinfo {author} {\bibfnamefont {C.}~\bibnamefont
  {Eurich}}, \bibinfo {author} {\bibfnamefont {J.}~\bibnamefont {Herrmann}}, \
  and\ \bibinfo {author} {\bibfnamefont {U.}~\bibnamefont {Ernst}},\
  }\href@noop {} {\bibfield  {journal} {\bibinfo  {journal} {Physical Review
  E}\ }\textbf {\bibinfo {volume} {66}},\ \bibinfo {pages} {066137} (\bibinfo
  {year} {2002})}\BibitemShut {NoStop}%
\bibitem [{\citenamefont {Bak}(1997)}]{bak1997nature}%
  \BibitemOpen
  \bibfield  {author} {\bibinfo {author} {\bibfnamefont {P.}~\bibnamefont
  {Bak}},\ }\href@noop {} {\emph {\bibinfo {title} {How nature works}}}\
  (\bibinfo  {publisher} {Oxford University Press, Oxford},\ \bibinfo {year}
  {1997})\BibitemShut {NoStop}%
\bibitem [{\citenamefont {Beggs}\ and\ \citenamefont
  {Plenz}(2003)}]{beggs2003neuronal}%
  \BibitemOpen
  \bibfield  {author} {\bibinfo {author} {\bibfnamefont {J.}~\bibnamefont
  {Beggs}}\ and\ \bibinfo {author} {\bibfnamefont {D.}~\bibnamefont {Plenz}},\
  }\href@noop {} {\bibfield  {journal} {\bibinfo  {journal} {the Journal of
  Neuroscience}\ }\textbf {\bibinfo {volume} {23}},\ \bibinfo {pages} {11167}
  (\bibinfo {year} {2003})}\BibitemShut {NoStop}%
\bibitem [{\citenamefont {Petermann}\ \emph {et~al.}(2009)\citenamefont
  {Petermann}, \citenamefont {Thiagarajan}, \citenamefont {Lebedev},
  \citenamefont {Nicolelis}, \citenamefont {Chialvo},\ and\ \citenamefont
  {Plenz}}]{petermann2009spontaneous}%
  \BibitemOpen
  \bibfield  {author} {\bibinfo {author} {\bibfnamefont {T.}~\bibnamefont
  {Petermann}}, \bibinfo {author} {\bibfnamefont {T.}~\bibnamefont
  {Thiagarajan}}, \bibinfo {author} {\bibfnamefont {M.}~\bibnamefont
  {Lebedev}}, \bibinfo {author} {\bibfnamefont {M.}~\bibnamefont {Nicolelis}},
  \bibinfo {author} {\bibfnamefont {D.}~\bibnamefont {Chialvo}}, \ and\
  \bibinfo {author} {\bibfnamefont {D.}~\bibnamefont {Plenz}},\ }\href@noop {}
  {\bibfield  {journal} {\bibinfo  {journal} {Proceedings of the National
  Academy of Sciences}\ }\textbf {\bibinfo {volume} {106}},\ \bibinfo {pages}
  {15921} (\bibinfo {year} {2009})}\BibitemShut {NoStop}%
\bibitem [{\citenamefont {Beggs}\ and\ \citenamefont
  {Plenz}(2004)}]{beggs2004neuronal}%
  \BibitemOpen
  \bibfield  {author} {\bibinfo {author} {\bibfnamefont {J.}~\bibnamefont
  {Beggs}}\ and\ \bibinfo {author} {\bibfnamefont {D.}~\bibnamefont {Plenz}},\
  }\href@noop {} {\bibfield  {journal} {\bibinfo  {journal} {The Journal of
  Neuroscience}\ }\textbf {\bibinfo {volume} {24}},\ \bibinfo {pages} {5216}
  (\bibinfo {year} {2004})}\BibitemShut {NoStop}%
\bibitem [{\citenamefont {Levina}\ \emph {et~al.}(2007)\citenamefont {Levina},
  \citenamefont {Herrmann},\ and\ \citenamefont
  {Geisel}}]{levina2007dynamical}%
  \BibitemOpen
  \bibfield  {author} {\bibinfo {author} {\bibfnamefont {A.}~\bibnamefont
  {Levina}}, \bibinfo {author} {\bibfnamefont {J.}~\bibnamefont {Herrmann}}, \
  and\ \bibinfo {author} {\bibfnamefont {T.}~\bibnamefont {Geisel}},\
  }\href@noop {} {\bibfield  {journal} {\bibinfo  {journal} {Nature Physics}\
  }\textbf {\bibinfo {volume} {3}},\ \bibinfo {pages} {857} (\bibinfo {year}
  {2007})}\BibitemShut {NoStop}%
\bibitem [{\citenamefont {Kinouchi}\ and\ \citenamefont
  {Copelli}(2006)}]{kinouchi2006optimal}%
  \BibitemOpen
  \bibfield  {author} {\bibinfo {author} {\bibfnamefont {O.}~\bibnamefont
  {Kinouchi}}\ and\ \bibinfo {author} {\bibfnamefont {M.}~\bibnamefont
  {Copelli}},\ }\href@noop {} {\bibfield  {journal} {\bibinfo  {journal}
  {Nature Physics}\ }\textbf {\bibinfo {volume} {2}},\ \bibinfo {pages} {348}
  (\bibinfo {year} {2006})}\BibitemShut {NoStop}%
\bibitem [{\citenamefont {Haldeman}\ and\ \citenamefont
  {Beggs}(2005)}]{haldeman2005critical}%
  \BibitemOpen
  \bibfield  {author} {\bibinfo {author} {\bibfnamefont {C.}~\bibnamefont
  {Haldeman}}\ and\ \bibinfo {author} {\bibfnamefont {J.}~\bibnamefont
  {Beggs}},\ }\href@noop {} {\bibfield  {journal} {\bibinfo  {journal}
  {Physical Review Letters}\ }\textbf {\bibinfo {volume} {94}},\ \bibinfo
  {pages} {58101} (\bibinfo {year} {2005})}\BibitemShut {NoStop}%
\bibitem [{\citenamefont {Uhlig}\ \emph {et~al.}(2013)\citenamefont {Uhlig},
  \citenamefont {Levina}, \citenamefont {Geisel},\ and\ \citenamefont
  {Herrmann}}]{uhlig2013critical}%
  \BibitemOpen
  \bibfield  {author} {\bibinfo {author} {\bibfnamefont {M.}~\bibnamefont
  {Uhlig}}, \bibinfo {author} {\bibfnamefont {A.}~\bibnamefont {Levina}},
  \bibinfo {author} {\bibfnamefont {T.}~\bibnamefont {Geisel}}, \ and\ \bibinfo
  {author} {\bibfnamefont {J.~M.}\ \bibnamefont {Herrmann}},\ }\href@noop {}
  {\bibfield  {journal} {\bibinfo  {journal} {Frontiers in Computational
  Neuroscience}\ }\textbf {\bibinfo {volume} {7}} (\bibinfo {year}
  {2013})}\BibitemShut {NoStop}%
\bibitem [{\citenamefont {Bertschinger}\ and\ \citenamefont
  {Natschl{\"a}ger}(2004)}]{bertschinger2004real}%
  \BibitemOpen
  \bibfield  {author} {\bibinfo {author} {\bibfnamefont {N.}~\bibnamefont
  {Bertschinger}}\ and\ \bibinfo {author} {\bibfnamefont {T.}~\bibnamefont
  {Natschl{\"a}ger}},\ }\href@noop {} {\bibfield  {journal} {\bibinfo
  {journal} {Neural Computation}\ }\textbf {\bibinfo {volume} {16}},\ \bibinfo
  {pages} {1413} (\bibinfo {year} {2004})}\BibitemShut {NoStop}%
\bibitem [{\citenamefont {Beggs}(2008)}]{beggs2008criticality}%
  \BibitemOpen
  \bibfield  {author} {\bibinfo {author} {\bibfnamefont {J.~M.}\ \bibnamefont
  {Beggs}},\ }\href@noop {} {\bibfield  {journal} {\bibinfo  {journal}
  {Philosophical Transactions of the Royal Society A: Mathematical, Physical
  and Engineering Sciences}\ }\textbf {\bibinfo {volume} {366}},\ \bibinfo
  {pages} {329} (\bibinfo {year} {2008})}\BibitemShut {NoStop}%
\bibitem [{\citenamefont {Newman}(2003)}]{newman2003structure}%
  \BibitemOpen
  \bibfield  {author} {\bibinfo {author} {\bibfnamefont {M.~E.}\ \bibnamefont
  {Newman}},\ }\href@noop {} {\bibfield  {journal} {\bibinfo  {journal} {SIAM
  review}\ }\textbf {\bibinfo {volume} {45}},\ \bibinfo {pages} {167} (\bibinfo
  {year} {2003})}\BibitemShut {NoStop}%
\bibitem [{\citenamefont {Holme}\ and\ \citenamefont
  {Kim}(2002)}]{holme2002growing}%
  \BibitemOpen
  \bibfield  {author} {\bibinfo {author} {\bibfnamefont {P.}~\bibnamefont
  {Holme}}\ and\ \bibinfo {author} {\bibfnamefont {B.~J.}\ \bibnamefont
  {Kim}},\ }\href@noop {} {\bibfield  {journal} {\bibinfo  {journal} {Physical
  Review E}\ }\textbf {\bibinfo {volume} {65}},\ \bibinfo {pages} {026107}
  (\bibinfo {year} {2002})}\BibitemShut {NoStop}%
\bibitem [{\citenamefont {Larremore}\ \emph {et~al.}(2011)\citenamefont
  {Larremore}, \citenamefont {Shew},\ and\ \citenamefont
  {Restrepo}}]{larremore2011predicting}%
  \BibitemOpen
  \bibfield  {author} {\bibinfo {author} {\bibfnamefont {D.~B.}\ \bibnamefont
  {Larremore}}, \bibinfo {author} {\bibfnamefont {W.~L.}\ \bibnamefont {Shew}},
  \ and\ \bibinfo {author} {\bibfnamefont {J.~G.}\ \bibnamefont {Restrepo}},\
  }\href@noop {} {\bibfield  {journal} {\bibinfo  {journal} {Physical Review
  Letters}\ }\textbf {\bibinfo {volume} {106}},\ \bibinfo {pages} {058101}
  (\bibinfo {year} {2011})}\BibitemShut {NoStop}%
\bibitem [{Note1()}]{Note1}%
  \BibitemOpen
  \bibinfo {note} {$\varphi _i(t)$ is undefined if node $i$ does not spike at
  $t$.}\BibitemShut {Stop}%
\bibitem [{Note2()}]{Note2}%
  \BibitemOpen
  \bibinfo {note} {In obtaining Eqn. \protect \textup {\hbox {\mathsurround \z@
  \protect \normalfont (\ignorespaces \ref {st_constraint}\unskip \@@italiccorr
  )}} we use the fact that if $x\leq a$ and $y\leq b$ then $x+y\leq
  a+b$.}\BibitemShut {Stop}%
\bibitem [{Note3()}]{Note3}%
  \BibitemOpen
  \bibinfo {note} {In obtaining these equations we have used the following
  rules for manipulating inequalities: if $x\leq a$ then $x+c\leq a+c$ for all
  $c$; if $x\leq a$ then $xc\leq ac$ if $c\geq 0$; if $x\leq a$ then $x/c\leq
  a/c$ if $c>0$. In applying the second of these rules we used $c=s(t+1)$,
  which is $\geq 0$. In applying the last of these rules we used $c=\protect
  \sqrt {\protect \mathcal {S}}$, which is $>0$.}\BibitemShut {Stop}%
\end{thebibliography}%

\end{document}